\documentclass[11pt]{article}

\newif\ifdraft
\draftfalse

\usepackage{fullpage}
\usepackage[numbers]{natbib}

\usepackage{algorithm, algorithmic}
\usepackage{xspace}
\usepackage{mathpartir}
\usepackage{graphicx}

\usepackage{amsmath, amssymb, amsthm}
\usepackage{verbatim}
\usepackage{color}
\usepackage{times}
\usepackage{xcolor}

\definecolor{DarkGreen}{rgb}{0.1,0.5,0.1}
\definecolor{DarkRed}{rgb}{0.5,0.1,0.1}
\definecolor{DarkBlue}{rgb}{0.1,0.1,0.5}

\usepackage[caption=false]{subfig}

\usepackage{url}

\usepackage[]{hyperref}
\hypersetup{
    unicode=false,          % non-Latin characters in Acrobat's bookmarks
    pdftoolbar=true,        % show Acrobat toolbar?
    pdfmenubar=true,        % show Acrobat menu?
    pdffitwindow=false,      % page fit to window when opened
    pdftitle={},    % title
    pdfauthor={}
    pdfsubject={},   % subject of the document
    pdfnewwindow=true,      % links in new window
    pdfkeywords={keywords}, % list of keywords
    colorlinks=true,       % false: boxed links; true: colored links
    linkcolor=DarkRed,          % color of internal links
    citecolor=DarkGreen,        % color of links to bibliography
    filecolor=DarkRed,      % color of file links
    urlcolor=DarkBlue,          % color of external links
}

\usepackage{cleveref}
\usepackage{thmtools, thm-restate}

\newcommand{\sw}[1]{\ifdraft \textcolor{blue}{[SW: #1]}\fi}
\newcommand{\ar}[1]{\ifdraft \textcolor{brown}{[AR: #1]}\fi}
\newcommand{\jh}[1]{\ifdraft \textcolor{red}{[JH: #1]}\fi}
\newcommand{\mg}[1]{\ifdraft \textcolor{purple}{[MG: #1]}\fi}
\newcommand{\eg}[1]{\ifdraft \textcolor{green}{[EG: #1]}\fi}

\newcommand\R{\mathbb{R}}
\newcommand\cA{\mathcal{A}}

\newcommand\cQ{\mathcal{Q}}

\newcommand\cX{\mathcal{X}}

\renewcommand{\tilde}{\widetilde}

\DeclareMathOperator*{\Expectation}{\mathbb{E}}
\newcommand{\Ex}[2]{\Expectation_{#1}\left[#2\right]}

\DeclareMathOperator*{\argmin}{\mathrm{argmin}}
\DeclareMathOperator*{\argmax}{\mathrm{argmax}}

\newcommand{\univ}{\cX}

\newcommand{\eps}{\varepsilon}

\renewcommand{\hat}{\widehat}
\renewcommand{\bar}{\overline}

\newcommand{\INDSTATE}[1][1]{\STATE\hspace{#1\algorithmicindent}}
\newcommand{\dq}{\mbox{{\sf DualQuery}}\xspace}
\newcommand{\cplex}{\mbox{{\sf CPLEX}}\xspace}
\newcommand{\MWEM}{\mbox{{\sf MWEM}}\xspace}

\newtheorem{theorem}{Theorem}[section]
\newtheorem{lemma}[theorem]{Lemma}

\newtheorem{remark}[theorem]{Remark}

\newtheorem{definition}[theorem]{Definition}

\theoremstyle{definition}

\title{Dual Query: Practical Private Query Release \\ for High
  Dimensional Data}

\author{
  Marco Gaboardi \qquad
  Emilio Jes\'us Gallego Arias \\
  Justin Hsu \qquad
  Aaron Roth \qquad
  Zhiwei Steven Wu
}

\begin{document}

\maketitle

\begin{abstract}
  We present a practical, differentially private algorithm for
  answering a large number of queries on high dimensional
  datasets. Like all algorithms for this task, ours necessarily has
  worst-case complexity exponential in the dimension of the
  data. However, our algorithm packages the computationally hard step
  into a concisely defined integer program, which can be solved
  non-privately using standard solvers. We prove accuracy and privacy
  theorems for our algorithm, and then demonstrate experimentally that
  our algorithm performs well in practice. For example, our algorithm
  can efficiently and accurately answer millions of queries on the
  Netflix dataset, which has over $17{,}000$ attributes; this is an
  improvement on the state of the art by multiple orders of magnitude.
\end{abstract}

\section{Introduction}
Privacy is becoming a paramount concern for machine learning and data analysis
tasks, which often operate on personal data. For just one example of
the tension between machine learning and data privacy, Netflix released an
anonymized dataset of user movie ratings for teams competing to develop an
improved recommendation mechanism. The competition was a great success (the
winning team improved on the existing recommendation system by more than
$10\%$), but the ad hoc anonymization was not as successful: \citet{NV08} were
later able to re-identify individuals in the dataset, leading to a
lawsuit and the cancellation of subsequent competitions.

{\em Differentially private query release} is an attempt to solve this problem.
Differential privacy is a strong formal privacy guarantee (that, among other
things, provably prevents re-identification attacks), and the problem of {\em
  query release} is to release accurate answers to a set of statistical queries.
As observed early on by \citet{BDMN05}, performing private query release is
sufficient to simulate any learning algorithm in the {\em statistical query
  model} of \citet{Kearns98}.

Since then, the query release problem has been extensively studied in the
differential privacy literature. While simple perturbation can
privately answer a small number of queries \citep{DMNS06}, more sophisticated
approaches can accurately answer nearly exponentially many queries in the size
of the private database \citep{BLR08,DNRRV09,DRV10,RR10,HR10,GRU12,HLM12}. A
natural approach, employed by many of these algorithms, is to
answer queries by generating \emph{synthetic data}: a safe version of the
dataset that approximates the real dataset on every statistical query of
interest.

Unfortunately, even the most efficient approaches for query release have a
per-query running time linear in the size of the \emph{data universe}---the
number of possible kinds of records. If each record is specified by a set of
attributes, the size of the data universe is exponential in the number of
attributes \citep{HR10}.  Moreover, this running
time is necessary in the worst case \citep{Ullman13,UV11}.

This exponential runtime has hampered practical evaluation of query release
algorithms. One notable exception is
due to \citet{HLM12}, who perform a thorough experimental evaluation of one
such algorithm, which they called \MWEM (Multiplicative Weights Exponential
Mechanism). They find that \MWEM has quite good accuracy in practice and scales
to higher dimensional data than suggested by a theoretical (worst-case)
analysis. Nevertheless, running time remains a problem, and the approach does
not seem to scale to high dimensional data (with more than $30$ or so attributes
for general queries, and more when the queries have more structure\footnote{%
  \citet{HLM12} are able to scale up to $1000$ features on synthetic data when the
  features are partitioned into a number of small buckets, and the queries
  depend on at most one feature per bucket.}).
The critical bottleneck is the size of the state maintained by the algorithm:
\MWEM, like many query release algorithms, needs to manipulate an object that has
size linear in the size of the data universe. This
quickly becomes impractical for records with even a modest number of attributes.

We present \dq, an alternative algorithm which is {\em dual} to \MWEM in a sense
that we will make precise.  Rather than manipulating an object of exponential
size, \dq solves a concisely represented (but NP-hard) optimization
problem.  Critically, the optimization step does not require a solution that is
private or exact, so it can be handled by existing, highly optimized
solvers. Except for this step, all parts of
our algorithm are efficient.  As a result, \dq requires (worst-case)
space and (in practice) time only linear in the number of {\em queries} of
interest, which is often significantly smaller than the size of the data
universe.  Like existing algorithms for query release, \dq has a provable
accuracy guarantee and satisfies the strong differential privacy guarantee. Both
\dq and \MWEM generate \emph{synthetic data}: the output of both algorithms is a
data set from the same domain as the input data set, and can thus be used as the
original data set would have been used. It is important to note that the
that the output data set is guaranteed to be
similar to the input data set only with respect to the query class; the
synthetic data might not be similar to the
real data in other respects.

% \ar{Removed discussion of the heuristic tuning. We could add it back in, but I
%thought it distracted in the intro. }
% \jh{Agreed.}
We evaluate \dq on a variety of datasets by releasing {\em 3-way marginals}
(also known as {\em conjunctions} or {\em contingency tables}), demonstrating
that it solves the query release problem accurately and efficiently even when
the data includes hundreds of thousands of features. We know of no other
algorithm that can perform accurate, private query release for this class of
queries on real data with more than even $100$ features.

\subsection*{Related work}
Differentially private learning has been studied since \citet{BDMN05} showed how
to convert learning algorithms in the \emph{SQ model} of \citet{Kearns98} into
differentially private learning algorithms with similar accuracy guarantees.
Since then, private machine learning has become a very active field with both
foundational sample complexity results \citep{KLNRS08,CH11,BNS13,DJW13} and
numerous efficient algorithms for particular learning problems
\citep{CM08,CMS11,RBHT09,KST12,CSS12,TS13}.

In parallel, there has been a significant amount of work on privately releasing
synthetic data based on a true dataset while preserving the answers to large
numbers of statistical queries \citep{BLR08,DNRRV09,RR10,DRV10,HR10,GRU12}.
These results are extremely strong in an information theoretic sense: they
ensure the consistency of the synthetic data with respect to an exponentially
large family of statistics. But, all of these algorithms (including the notable
multiplicative weights algorithm of \citet{HR10}, which achieves the
theoretically optimal accuracy and runtime) have running time exponential in the
dimension of the data.  Under standard cryptographic assumptions, this is
necessary in the worst case for mechanisms that answer arbitrary statistical
queries \citep{Ullman13}.

Nevertheless, there have been some experimental evaluations of these approaches
on real datasets. Most related to our work is the evaluation of the \MWEM
mechanism by \citet{HLM12}, which is based on the private multiplicative weights
mechanism \citep{HR10}. This algorithm is inefficient---it manipulates a
probability distribution over a set exponentially large in the dimension of the
data space---but with some heuristic optimizations, \citet{HLM12} were able to
implement the multiplicative weights algorithm on real datasets with up
to 77 attributes (and even more when the queries are restricted to take positive
values only on a small number of disjoint groups of features). However, it seems
difficult to scale this approach to higher dimensional data.

Another family of query release algorithms are based on the Matrix Mechanism
\citep{CHRMM10,LM12}. The runtime guarantees of the matrix mechanism are
similar to the approaches based on multiplicative weights---the algorithm
manipulates a ``matrix'' of queries with dimension exponential in the number of
features.  \citet{CPSY12} evaluate an approach based on this family of
algorithms on low dimensional datasets, but scaling to high dimensional data
also seems challenging. A recent work by \citet{PrivBayes} proposes
a low-dimensional approximation for high-dimensional data distribution by
privately constructing Bayesian networks, and shows that such a
representation gives good accuracy on some real datasets.

Our algorithm is inspired by the view of the synthetic data generation problem
as a zero-sum game, first proposed by \citet{HRU13}. In this interpretation,
\citet{HLM12} solves the game by having a {\em data player} use a no-regret
learning algorithm, while the {\em query player} repeatedly best responds by
optimizing over queries. In contrast, our algorithm swaps the roles of the two
players: the query player now uses the no-regret learning algorithm, whereas the
data player now finds best responses by solving an optimization problem. This is
reminiscent of ``Boosting for queries,'' proposed by \citet{DRV10}; the
main difference is that our optimization problem is over single records
rather than sets of records. As a result, our optimization can be
handled non-privately.

There is also another theoretical approach to query release due to Nikolov,
Talwar, and Zhang (and later specialized to marginals by Dwork, Nikolov, and
Talwar) that shares a crucial property of the one we present here---namely that
the computationally difficult step does not need to be solved privately
\citep{NTZ13,DNT14}. The benefit of our approach is that it yields synthetic data
rather than just query answers, and that the number of calls to the optimization
oracle is smaller. However, the approach of \citet{NTZ13,DNT14} yields
theoretically optimal accuracy bounds (ours does not), and so that approach
certainly merits future empirical evaluation.

\section{Differential privacy background}
Differential privacy has become a standard algorithmic notion
for protecting the privacy of individual records in a statistical
database. It formalizes the requirement that the addition or removal
of a data record does not change the probability of any outcome of the
mechanism by much.

To begin, databases are multisets of elements from an abstract domain
$\cX$, representing the set of all possible data records.  Two databases $D,
D'\subset \cX$ are {\em neighboring} if they differ in a single data element
($| D \triangle D'| \leq 1$).
% \ar{Confusing -- $D - D'$ is not the notation for symmetric difference of sets,
%   and if we write $D \in X^n$, we have not expressed databases as histograms (or
%   sets -- here databases are ordered strings). We should pick some consistent
%   representation -- perhaps sets is clearest here. Then we can just say
%   databases $D \subset X$ are multisets of elements.}
% \sw{fixed?}
% \ar{Looks good.}

\begin{definition}[\citet{DMNS06}] A mechanism
$M\colon \cX^n \rightarrow R$ satisfies $(\eps, \delta)$-differential
privacy if for every $S\subseteq R$ and for all neighboring databases
$D, D'\in \cX^n$, the following holds:
\[
\Pr[M(D)\in S] \leq e^\eps \Pr[M(D') \in S] + \delta
\]
If $\delta =0$ we say $M$ satisfies $\eps$-differential privacy.
\end{definition}

\begin{definition}
The \emph{(global) sensitivity} of a query $q$ is its maximum difference
when evaluated on two neighboring databases:
\[
GS_f = \max_{D, D' \in \cX^n : |D\triangle D'| = 1} |q(D) - q(D')|.
\]
\end{definition}

In this paper, we consider the private release of information for the
classes of linear queries, which have sensitivity $1/n$.

\begin{definition}
  For a predicate $\varphi\colon \cX\rightarrow \{0,1\}$, the \emph{linear
  query} $q_\varphi\colon \cX^n \rightarrow [0,1]$ is defined by
\[
  q_\varphi(D) = \frac{\sum_{x\in D} \varphi(x)}{|D|} .
\]
We will often represent linear queries a different form, as a vector $q \in \{
0, 1\}^{|\cX|}$ explicitly encoding the predicate $\phi$:
\[
  q(D) = \frac{\sum_{x \in D} q_x}{|D|} .
\]
\end{definition}
%\ar{Here we seem to be using the set representation.}

We will use a fundamental tool for private data analysis: we can bound the
privacy cost of an algorithm as a function of the privacy costs of its
subcomponents.

\begin{lemma}[\citet{DRV10}] \label{composition}
Let $M_1,\ldots , M_k$ be such that each $M_i$ is $(\eps_i,0)$-private with
$\eps_i\leq \eps'$.  Then $M(D) = (M_1(D), \ldots , M_k(D))$ is
$(\eps,0)$-private for
 $ \eps = \sum_{i = 1}^k \eps_i,$
and $(\eps, \delta)$-private for
\[
  \eps = \sqrt{2 \log(1 / \delta) k}\eps' + k\eps'(e^{\eps'} - 1)
\]
for any $\delta \in (0, 1)$. The sequence of mechanisms can be chosen
adaptively, i.e., later mechanisms can take outputs from previous mechanisms as
input.

\end{lemma}

% \jh{Should probably define sensitivity if we are going to mention it.
%   Alternative, can vaguely claim that techniques based on simple perturbation
%   are trivial for $k \geq n$.}

% \sw{TODO: Remark that in our experiments, the number of queries we are
%   answering is at least as large as the size of the database. Laplace
%   mechanism would give trivial accuracy. }

% \jh{One more thing that we will need: composition theorems (standard and
%   advanced). Probably no need to define the adaptive composition experiment and
%   all that, a simple version should be enough.}

% \sw{should I cite the DMNS06 or the Boosting paper for the composition?}

% \ar{Cite the boosting paper -- DRV10}

\section{The query release game} \label{sec:qr:game}
The analysis of our algorithm relies on the interpretation of query release as a two
player, zero-sum game~\citep{HRU13}. In the present section, we review this idea
and related tools.

\subsection*{Game definition}
Suppose we want to answer a set of queries $\cQ$. For each query $q \in \cQ$, we
can form the {\em negated query} $\bar{q}$, which takes values $\bar{q}(D) = 1 -
q(D)$ for every database $D$. %
Equivalently, for a linear query defined by a
predicate $\varphi$, the negated query is defined by the negation
$\neg \varphi$ of the predicate.
% \ar{Maybe its clearer to define the negated query in terms of universe elements?}
For the remainder, we will assume that $\cQ$ is closed under negation; if not,
we may add negated copies of each query to $\cQ$.

Let there be two players, whom we call the {\em data} player and {\em query}
player. The data player has action set equal to the data universe $\cX$, while
the query player has action set equal to the query class $\cQ$. Given a play $x
\in \cX$ and $q \in \cQ$, we let the payoff be
\begin{equation} \label{eq:payoff}
  A(x, q) := q(D) - q(x),
\end{equation}
where $D$ is the true database.  To define the zero-sum game, the data player
will try to minimize the payoff, while the query player will try to maximize the
payoff.

% \ar{Removed paragraph about intuition, which I think was confusing and slightly
%   wrong -- since we don't have absolute values, the payoff for playing q, x is
%   not quite the query-evaluation error...}
% \jh{Slightly wrong intuition is the best kind of intuition.}
%For intuition, think of
%the data player as trying to find a good approximation to the database, but with
%only one kind of record. The payoff, then, is the error of this approximation on
%query $q$, and the data player is trying to minimize error.  Likewise, the query
%player tries to find ``bad'' queries, i.e., queries that are poorly handled. So,
%the query player tries to maximize the payoff (error).

\subsection*{Equilibrium of the game}

Let $\Delta(\cX)$ and $\Delta(\cQ)$ be the set of probability distributions over
$\cX$ and $\cQ$.  We consider how well each player can do if they randomize over
their actions, i.e., if they play from a probability distribution over their
actions.  By von Neumann's minimax theorem,
\[
  \min_{u \in \Delta(\cX)} \max_{w \in \Delta(\cQ)} A(u, w) =
  \max_{w \in \Delta(\cQ)} \min_{u \in \Delta(\cX)} A(u, w),
\]
for any two player zero-sum game, where
\[
  A(u, w) := \mathbb{E}_{x \sim u, q \sim w} A(x, q)
\]
is the expected payoff.  The common value is called the {\em value of the game},
which we denote by $v_A$.

Intuitively, von Neumann's theorem states that there is no advantage in a
player going first: the minimizing player can always force payoff at most $v_A$,
while the maximizing player can always force payoff at least $v_A$.
This suggests that each player can play an optimal strategy, assuming best play
from the opponent---this is the notion of equilibrium strategies, which we now
define. We will soon interpret these strategies as solutions to the query
release problem.

\begin{definition}
  Let $\alpha > 0$. Let $A$ be the payoffs for a two player, zero-sum game with
  action sets $\cX, \cQ$. Then, a pair of strategies $u^* \in \Delta(\cX)$
  and $w^* \in \Delta(\cQ)$ form an {\em $\alpha$-approximate mixed Nash
    equilibrium} if
  \[
    A(u^*, w) \leq v_A + \alpha
    \qquad \text{and} \qquad
    A(u, w^*) \geq v_A - \alpha
  \]
  for all strategies $u \in \Delta(\cX)$ and $w \in \Delta(\cQ)$.
\end{definition}

If the true database $D$ is normalized to be a distribution $\hat{D}$ in
$\Delta(\cX)$, then $\hat{D}$ always has zero payoff:
\[
  A(\hat{D}, w) = \mathbb{E}_{x \sim \hat{D}, q \sim w} [q(D) - q(x)] = 0.
\]
Hence, the value of the game $v_A$ is at most $0$. Also, for any
data strategy $u$, the payoff of query $q$ is the negated payoff of the negated
query $\bar{q}$:
\begin{align*}
  A(u, \bar{q}) &= \mathbb{E}_{x \sim u} [\bar{q}(D) - \bar{q}(x)]
  = \mathbb{E}_{x \sim u} [q(x) - q(D)],
\end{align*}
which is $A(u, \bar{q})$.  Thus, any query strategy that places equal weight on
$q$ and $\bar{q}$ has expected payoff zero, so $v_A$ is at least $0$. Hence,
$v_A = 0$.

Now, let $(u^*, w^*)$ be  an $\alpha$-approximate equilibrium. Suppose that the
data player plays $u^*$, while the query player always plays query $q$.  By the
equilibrium guarantee, we then have $A(u^*, q) \leq \alpha$, but the expected
payoff on the left is simply $q(D) - q(u^*)$.  Likewise, if the query player
plays the negated query $\bar{q}$, then
\[
  -q(D) + q(u^*) = A(u^*, \bar{q}) \leq \alpha,
\]
so $q(D) - q(u^*) \geq -\alpha$. Hence for every query $q \in \cQ$, we know
$|q(u^*) - q(D)| \leq \alpha$. This is precisely what we need for query release:
we just need to privately calculate an approximate equilibrium.

\subsection*{Solving the game}

To construct the approximate equilibrium, we will use the multiplicative weights update
algorithm (MW).\footnote{%
  The MW algorithm has wide applications; it has been rediscovered in various
  guises several times.  More details can be found in the comprehensive survey
  by \citet{AHK12}.}
This algorithm maintains a distribution over actions (initially uniform) over a
series of steps. At each step, the MW algorithm receives a (possibly adversarial)
loss for each action.  Then, MW reweights the distribution to favor actions with
less loss.
The algorithm is presented in \Cref{alg:MW}.
\begin{algorithm}[h!]
  \begin{algorithmic}
    \STATE{Let $\eta > 0$ be given, let $\cA$ be the action space}
    \STATE{Initialize $\tilde{A}^1$ uniform distribution on $\cA$}
    \STATE{For $t = 1,2,\dots,T$:}
    \INDSTATE[1]{Receive loss vector $\ell^t$}
    \INDSTATE[1]{{\bf For each} $a \in \cA${\bf :}}
    \INDSTATE[2]{Update $A^{t+1}_a = e^{- \eta \ell^t_a} \tilde{A}^{t}_a$ for every
  $a \in \cA$}
  \INDSTATE[1]{Normalize $\tilde{A}^{t+1} = \frac{A^{t + 1}}{\sum_i A^{t + 1}_i}$}
  \end{algorithmic}
  \caption{The Multiplicative Weights Algorithm}
  \label{alg:MW}
\end{algorithm}

For our purposes, the most important application of MW is to solving zero-sum
games. \citet{FS96} showed that if one player maintains a distribution over
actions using MW, while the other player selects a {\em best-response} action
versus the current MW distribution (i.e., an action that maximizes his expected
payoff), the average MW distribution and empirical best-response distributions
will converge to an approximate equilibrium rapidly.

\begin{theorem}[\citet{FS96}]
  \label{thm:mw-eq}
  Let $\alpha > 0$, and let $A(i, j) \in [-1, 1]^{m \times n}$ be the payoff
  matrix for a zero-sum game. Suppose the first player uses multiplicative
  weights over their actions to play distributions $p^1, \dots, p^T$, while
  the second player plays $(\alpha/2)$-approximate best responses $x^1, \dots,
  x^T$, i.e.,
  \[
    A(x^t, p^t) \geq \max_x A(x, p^t) - \alpha/2.
  \]
  Setting $T = 16 \log n /\alpha^2$ and $\eta = \alpha/4$ in the MW algorithm,
  the empirical distributions
  \[
    \frac{1}{T}\sum_{i = 1}^T p^i
    \quad \text{and} \quad
    \frac{1}{T} \sum_{i = 1}^T x^i
  \]
  form an $\alpha$-approximate mixed Nash equilibrium.
\end{theorem}
%\begin{proof}
%  Almost direct from Theorem $3.1$ in \citet{AHK12}.
%\end{proof}

\section{Dual query release}
Viewing query release as a zero-sum game, we can interpret the algorithm of
\citet{HR10} (and the \MWEM algorithm of \citet{HLM12}) as solving the game by
using MW for the data player, while the query player plays best responses. To
guarantee privacy, their algorithm selects the query best-responses privately
via the exponential mechanism of \citet{MT07}. Our algorithm simply reverses
the roles: while \MWEM uses a no-regret algorithm to maintain the data player's
distribution, we will instead use a no-regret algorithm for the query player's
distribution; instead of finding a maximum payoff query at each round, our
algorithm selects a minimum payoff record at each turn.

In a bit more detail, we maintain a distribution $\cQ$ over the queries,
initially uniform. At each step $t$, we first simulate the data player making a
best response against $\cQ$, i.e., selecting a record $x^t$ that maximizes the
expected payoff if the query player draws a query according to $\cQ$. As we
discuss below, for privacy considerations we cannot work directly with the
distribution $\cQ$. Instead, we sample $s$ queries $\{ q_i \}$ from $\cQ$, and
form ``average'' query $\tilde{q}$, which we take as an estimate of the true
distribution $\cQ$. The data player then best-responds against $\tilde{q}$.

Then, we simulate the query player updating the distribution $\cQ$ after seeing
the selected record $x^t$---we will reduce weight on queries that have similar
answers on $x^t$ the database $D$, and we will increase weight on queries that
have very different answers on $x^t$ and $D$. After running for $T$ rounds, we
output the set of records $\{ x_t \}_t$ as the synthetic database. Note that this is a data set from the same domain as the original data set, and so can be used in any application that the original data set can be used in -- with the caveat of course, that except with respect to the query class used in the algorithm, there are no guarantees about how the synthetic data resembles the real data.  The full
algorithm can be found in \Cref{alg:dualquery}; the main parameters are $\alpha
\in (0, 1)$, which is the target maximum additive error of \emph{any} query on
the final output, and $\beta \in (0, 1)$, which is the probability that the
algorithm may fail to achieve the accuracy level $\alpha$. These parameters two
control the number of steps we take ($T$), the update rate of the query
distribution ($\eta$), and the number of queries we sample at each step ($s$).

%We argue that our dual algorithm is private in a very different way from the
%more straightforward privacy argument
Our privacy argument differs slightly from the analysis for \MWEM. There, the
data distribution is public, and finding a query with high error requires access
to the private data.  Our algorithm instead maintains a distribution $Q$ over
queries which depends directly on the private data, so we cannot use $Q$
directly. Instead, we argue that \emph{queries sampled from $Q$} are privacy
preserving. Then, we can use a non-private optimization method to find a
minimal error record versus queries sampled from $Q$.  We then trade off
privacy (which degrades as we take more samples) with accuracy (which improves
as we take more samples, since the distribution of sampled queries converges to
$Q$).

Given known hardness results for the query release problem \citep{Ullman13}, our
algorithm must have worst-case runtime polynomial in the universe size $|\cX|$,
so is not theoretically more efficient than prior approaches.  In fact, even
compared to prior work on query release (e.g., \citet{HR10}), our algorithm has
a weaker accuracy guarantee.  However, our approach has an important practical
benefit: the computationally hard step can be handled with standard, non-private
solvers (we note that this is common also to the approach of \cite{NTZ13,DNT14}).

% \jh{Tried to improve. Also, this last paragraph above is a bit strong...}
% \sw{It's a bit strong; and we should merge this with the previous two}
The iterative structure of our algorithm, combined with our use of constraint
solvers, also allows for several heuristics improvements. For instance, we may run
for fewer iterations than predicted by theory. Or, if the optimization problem
turns out to be hard (even in practice), we can stop the solver early at a
suboptimal (but often still good) solution.  These heuristic tweaks can improve
accuracy beyond what is guaranteed by our accuracy theorem, while always
maintaining a strong \emph{provable} privacy guarantee.

\begin{algorithm}
  \begin{algorithmic}
    \STATE{\textbf{Parameters:} Target accuracy level $\alpha \in (0, 1)$,
      target failure probability $\beta \in (0, 1)$.}
    \STATE{\textbf{Input:}
      Database $D \in \R^{|\univ|}$ (normalized) and linear queries
      $q_1, \dots, q_k \in \{0, 1\}^{|\univ|}$.}
    \STATE{\textbf{Initialize:}
      Let $\cQ = \bigcup_{j = 1}^{k} q_j \cup \bar{q_j}$,
      $Q^1$ be a uniform distribution on $\cQ$,
      \[
        T  = \frac{16 \log |\cQ|}{\alpha^2},
        \qquad
        \eta = \frac{\alpha}{4},
        \qquad \text{and} \qquad
        s = \frac{48 \log \left({2 |\univ| T}/{\beta} \right)}{\alpha^2} .
      \]
      Let the payoff function $A_D : \cX \times [0, 1]^{|\cX|} \to \R$ be:
      \[
        A_D(x, \tilde{q}) = \tilde{q}(D) - \tilde{q}(x) ,
        \qquad
        \text{where}
        \qquad
        \tilde{q}(D) = \frac{1}{|D|} \sum_{x \in D} \tilde{q}_x .
      \]
    }
    \STATE{For $t = 1,\dots,T$:}
    \INDSTATE[1]{Sample $s$ queries $\{q_i\}$ from $\cQ$ according to
    $Q^{t}$.}
    \INDSTATE[1]{Let $\tilde{q} := \frac{1}{s} \sum_i q_i$.}
    \INDSTATE[1]{Find $x^t$ with $A_D(x^t, \tilde{q}) \geq \max_x
    A_D(x, \tilde{q}) - \alpha/4$.}
    \INDSTATE[1]{\textbf{Update:} For each $q \in \cQ$:}
    \INDSTATE[2]{$Q^{t+1}_q := \exp(-\eta A_D(x^t, q) \rangle)
      \cdot Q^{t}_q.$}
    \INDSTATE[1]{Normalize $Q^{t+1}$.}

    \STATE{Output synthetic database $\hat{D} := \bigcup_{t = 1}^T x^t.$}
 \end{algorithmic}
 \caption{\dq}
 \label{alg:dualquery}
\end{algorithm}

\subsection*{Privacy}
The privacy proofs are largely routine, based on the composition theorems.
Rather than fixing $\eps$ and solving for the other parameters as is typical in
the literature, we present the
privacy cost $\eps$ as function of parameters $T, s, \eta$. This form will lead
to simpler tuning of the parameters in our experimental evaluation.

We will use the privacy of the following mechanism (result due to \citet{MT07})
as an ingredient in our privacy proof.

\begin{lemma}[\citet{MT07}]
  \label{def:expmech}
  Given some arbitrary output range $R$, the {\em exponential mechanism} with
  score function $S$ selects and outputs an element $r\in R$ with probability
  proportional to
  \[
    \exp\left(\frac{\eps S(D, r)}{2\cdot GS_S}\right),
  \]
  where $GS_S$ is the {\em sensitivity} of $S$, defined as
  \[
    GS_S = \max_{D, D': |D\triangle D'| = 1, r\in R} |S (D, r) - S (D', r)|.
  \]
  The exponential mechanism is $\eps$-differentially private.
\end{lemma}
% \ar{Do we need to define the exponential mechanism here? We never really use it,
%   except by analogy when showing that samples from the MW distribution are
%   private---but perhaps we could bring this up (if at all) in the privacy
%   proof.}
% \jh{We can definitely move this later. However, it seems like the most direct
%   way to prove privacy (rather than from scratch), yes?}
% \sw{I guess we can simply quote exponential mechanism when we do the
%   privacy proof?}

We first prove pure $\eps$-differential privacy.

\begin{theorem}
  \label{thm:privacy2}
  \dq is $\eps$-differentially private for
  \[
    \eps = \frac{\eta T(T-1)s}{n}.
  \]
\end{theorem}

\begin{proof}
  We will argue that sampling from $Q^t$ is equivalent to running the
  exponential mechanism with some quality score.  At round $t$, let $\{x^i\}$
  for $i \in [t-1]$ be the best responses for the previous rounds. Let $r(D, q)$
  be defined by
  \[
    r(D, q) = \sum_{i = 1}^{t - 1} ( q(D) - q(x^i) ),
  \]
% \sw{subscript $i$ should range from 1 to $t-1$? In the first round, the
%   mechanism should be 0-private. For large $\eta$ and $s$, this could make some
%   difference.}
% \jh{Sure.}
  where $q \in \cQ$ is a query and $D$ is the true database. This function is
  evidently $((t-1)/n)$-sensitive in $D$: changing $D$ changes each $q(D)$ by
  at most $1/n$.  Now, note that sampling from $Q^t$ is simply sampling from the
  exponential mechanism, with quality score $r(D, q)$. Thus, the privacy cost of
  each sample in round $t$ is $\eps_t' = 2\eta (t-1)/n$ by \Cref{def:expmech}.

  By the standard composition theorem (\Cref{composition}), the total privacy cost is
  \[
    \eps = \sum_{t = 1}^{T} s\eps_t' = \frac{2 \eta s}{n} \cdot
    \sum_{t = 1}^{T} (t-1) = \frac{\eta T(T - 1)s}{n}.
  \]
\end{proof}

We next show that \dq is $(\eps, \delta)$-differentially private, for a much
smaller $\eps$.

% \sw{We should be doing the same thing in $(\eps, \delta)$ version too:
%   treating each sampling as a mechanism with privacy cost
%   $2\eta(T-1)/n$. Also, we can ignore the first round as the samplings
% from the initial uniform distribution is $0$-private. Might be able to
% gain something here for our experiments. }
% \jh{Just to clarify, this has been updated, yes?}\sw{yes}

\begin{theorem}
  \label{thm:privacy}
  Let $0 < \delta < 1$. \dq is $(\eps, \delta)$-differentially private for
  \[
    \eps = \frac{2 \eta (T-1)}{n} \cdot \left[ \sqrt{2s(T-1) \log(1/\delta)} + s(T-1) \left(
        \exp\left( \frac{2 \eta (T-1)}{n} \right) - 1 \right) \right].
  \]
\end{theorem}
\begin{proof}
  Let $\eps$ be defined by the above equation. By the advanced composition
  theorem (\Cref{composition}), running a composition of $k$ $\eps'$-private
  mechanisms is $(\eps, \delta)$-private for
  \[
    \eps = \sqrt{2k \log(1/\delta)} \eps' + k \eps' (\exp(\eps') -
    1).
  \]
  Again, note that sampling from $Q^t$ is simply sampling from the exponential
  mechanism, with a $(T-1)/n$-sensitive quality score. Thus, the privacy cost of
  each sample is $\eps' = 2\eta (T-1)/n$ by \Cref{def:expmech}. We plug in $k
  = s(T-1)$ samples, as in the first round our samplings are $0$-differentially
  private.
\end{proof}

\subsection*{Accuracy}

The accuracy proof proceeds in two steps. First, we show that ``average query''
formed from the samples is close to the true weighted distribution $Q^t$. We
will need a standard Chernoff bound.

\begin{lemma}[Chernoff bound] \label{chernoff}
  Let $X_1, \dots, X_N$ be IID random variables with mean $\mu$, taking
  values in $[0, 1]$.  Let $\bar{X} = \frac{1}{N} \sum_i X_i$ be the empirical
  mean.  Then,
  \[
    \Pr[ |\bar{X} - \mu| > T ] < 2 \exp (-NT^2/3)
  \]
  for any $T$.
\end{lemma}

\begin{lemma}
  \label{lem:sampling}
  Let $\beta \in (0, 1)$, and let $p$ be a distribution over queries. Suppose we draw
  \[
    s = \frac{48 \log \left(\frac{2 |\univ|}{\beta} \right)}{\alpha^2}
  \]
  samples $\{\hat{q_i}\}$ from $p$, and let $\bar{q}$ be the aggregate query
  \[
    \bar{q} = \frac{1}{s} \sum_{i = 1}^s \hat{q_i}.
  \]
  Define the true weighted answer $Q(x)$ to be
  \[
    Q(x) = \sum_{i = 1}^{|\cQ|} p_i q_i(x).
  \]
  Then with probability at least $1 - \beta$, we have $|\bar{q}(x) - Q(x) | <
  \alpha/4$ for every $x \in \cX$.
\end{lemma}
\begin{proof}
  For any fixed $x$, note that $\bar{q}(x)$ is the average of random variables
  $\hat{q_1}(x), \dots ,\hat{q_s}(x)$. Also, note that $\mathbb{E}[\bar{q}(x)] =
  Q(x)$.  Thus, by the Chernoff bound (\Cref{chernoff}) and our choice of $s$,
  \[
    \Pr[ |\bar{q}(x) - Q(x)| > \alpha/4 ] < 2 \exp (-s\alpha^2/48) =
    \beta / |\univ|.
  \]
  By a union bound over $x \in \cX$, this equation holds for all $x \in \cX$
  with probability at least $1 - \beta$.
\end{proof}

Next, we show that we compute an approximate equilibrium of the query release
game. In particular, the record best responses form a synthetic database that
answer all queries in $\cQ$ accurately. Note that our algorithm doesn't require
an exact best response for the data player; an approximate best response will
do.
\ar{I'm a bit confused by the following theorem -- I guess $\eps$ and
  $\delta$ are the privacy parameters associated with the algorithm parameters
  specified in the algorithm -- but we haven't actually stated that the
  algorithm is (eps,delta) private yet -- maybe we should state that as a
  corollary of the privacy theorem.}
\jh{Fixed, hopefully. Do you think we should present the remark in the short
  version?}
% \jh{Not sure what you mean. Things are a bit circular maybe, but we doesn't
%   \Cref{thm:privacy} show $(\eps, \delta)$-privacy?}
% \jh{Note to self: circular is bad.}

\begin{theorem}
  \label{thm:utility}
  With probability at least $1 - \beta$, \dq finds a synthetic database that
  answers all queries in $\cQ$ within additive error
  $\alpha$.
\end{theorem}
% \sw{we have been mentioning worse guarantee. but this theorem doesn't
%   show the dependence}

\begin{proof}
  As discussed in \Cref{sec:qr:game}, it suffices to show that the distribution
  of best responses $x^1, \dots, x^T$ forms is an $\alpha$-approximate
  equilibrium strategy in the query release game. First, we set the number of
  samples $s$ according to in \Cref{lem:sampling} with failure probability
  $\beta / T$.  By a union bound over $T$ rounds, sampling is successful for
  every round with probability at least $1 - \beta$; condition on this event.

  Since we are finding an $\alpha/4$ approximate best response to the sampled
  aggregate query $\bar{q}$, which differs from the true distribution by at most
  $\alpha/4$ (by \Cref{lem:sampling}), each $x^i$ is an $\alpha/4 + \alpha/4 =
  \alpha/2$ approximate best response to the true distribution $Q^t$. Since $q$
  takes values in $[0, 1]$, the payoffs are all in $[-1, 1]$. Hence,
  \Cref{thm:mw-eq} applies; setting $T$ and $\eta$ accordingly gives the result.
\end{proof}

\begin{remark}
  The guarantee in \Cref{thm:utility} may seem a little unusual, since the
  convention in the literature is to treat $\eps, \delta$ as inputs to the
  algorithm. We can do the same: from \Cref{thm:privacy} and plugging in for $T,
  \eta, s$, we have
  \begin{align*}
    \eps = \frac{4\eta T \sqrt{2 sT \log(1/\delta)}}{n}
    = \frac{256 \log^{3/2} |\cQ| \sqrt{6 \log (1/\delta)
        \log(2|\cX|T/\beta)}}{\alpha^3 n}.
  \end{align*}
  Solving for $\alpha$, we find
  \[
    \alpha = O \left( \frac{\log^{1/2}|\cQ| \log^{1/6} (1/\delta)
        \log^{1/6}(2|\cX|/\gamma)}{n^{1/3} \eps^{1/3}} \right),
  \]
  for $\gamma < \beta/T$.
\end{remark}
% \fi

\section{Case study: 3-way marginals}
% \ar{Removing reference to parities, and parity section}

In our algorithm, the computationally difficult step is finding the data
player's approximate best response against the query player's distribution. As
mentioned above, the form of this problem depends on the particular query class
$\cQ$. In this section, we first discuss the optimization problem in general,
and then specifically for the well-studied class of \emph{marginal} queries.  For instance, in a database of medical information
in binary attributes, a particular marginal query may be: What fraction of the
patients are over 50, smoke, and exercise?

\subsection*{The best-response problem}
Recall that the query release game has payoff $A(x, q)$ defined by
\Cref{eq:payoff}; the data player tries to minimize the payoff, while the query
player tries to maximize it. If the query player has distribution $Q^t$ over
queries, the data player's best response minimizes the expected loss:
\[
  \argmin_{x \in \cX} \Ex{q \leftarrow Q^t}{q(D) - q(x)}.
\]

To ensure privacy, the data player actually plays against the distribution of
samples $\hat{q_1}, \dots, \hat{q_s}$. Since the database $D$ is fixed and
$\hat{q_i}$ are linear queries, the best-response problem is
\[
  \argmin_{x \in \cX} \frac{1}{s} \sum_{i = 1}^s \hat{q_i}(D) - \hat{q_i}(x)
  = \argmax_{x\in \cX} \sum_{i=1}^s \hat{q_i}(x).
%  = \argmax_{x \in \cX} \bar{q}(x),
\]
% since $D$ is fixed and $\hat{q_i}$ are linear queries.
%  (As before, the aggregate
% query $\bar{q}$ is the average of the sampled queries $\hat{q_i}$.)
% This
% optimization problem can be solved with no regard to privacy, since the queries
%$\hat{q_i}$ and $\bar{q}$ can safely be released.
%In fact,

By \Cref{thm:utility} it even suffices to find an approximate maximizer, in
order to guarantee accuracy.

\subsection*{3-way marginal queries}
To look at the precise form of the best-response problem, we consider {\em
  3-way marginal} queries. We think of records as having $d$ binary attributes,
so that the data universe $|\cX|$ is all bitstrings of length $d$. We write
$x_i$ for $x \in \cX$ to mean the $i$th bit of record $x$.
\begin{definition}
  Let $\cX = \{0, 1\}^d$. A {\em 3-way marginal query} is a linear query
  specified by 3 integers $a \neq b \neq c \in [d]$,
  taking values
  \[
    q_{abc}(x) = \left\{
      \begin{array}{ll}
        1 &: x_a = x_b = x_c = 1 \\
        0 &: \text{otherwise.}
      \end{array}
    \right.
  \]
\end{definition}

Though everything we do will apply to general $k$-way marginals, for
concreteness we will consider $k = 3$: 3-way marginals.

Recall that the query class $\cQ$ includes each query and its negation.  So, we also
have negated conjunctions:
  \[
    \bar{q_{abc}}(x) = \left\{
      \begin{array}{ll}
        0 &: x_a = x_b = x_c = 1 \\
        1 &: \text{otherwise.}
      \end{array}
    \right.
  \]

Given sampled conjunctions $\{\hat{u_i}\}$ and negated conjunctions
$\{\hat{v_i}\}$, the best-response problem is
\[
  \argmax_{x \in \cX} \sum_i \hat{u_i}(x) + \sum_j \hat{v_j}(x).
\]
In other words, this is a MAXCSP problem---we can associate a clause to each
conjunction:
\[
  q_{abc} \Rightarrow  (x_a \wedge x_b \wedge x_c)
  \quad \text{and} \quad
  \bar{q_{abc}} \Rightarrow   (\bar{x_a} \vee \bar{x_b} \vee \bar{x_c}),
\]
and we want to find $x \in \{0, 1\}^d$ satisfying as many clauses as
possible.\footnote{%
  Note that this is almost a MAX3SAT problem, except there are also
  ``negative'' clauses.}\eg{I wonder about the footnote...}

Since most solvers do not directly handle MAXCSP problems, we convert this
optimization problem into a more standard, integer program form. We introduce a
variable $x_i$ for each literal $x_i$, a variable $c_i$ for each sampled query,
positive or negative.  Then, we form the following integer program encoding to
the best-response MAXCSP problem.
\begin{align*}
  \max &\sum_i c_i & \\
  \text{such that } & x_a + x_b + x_c \geq 3c_i &\text{ for each } \hat{u_i} = q_{abc} \\
  & (1 - x_a) + (1 - x_b) + (1 - x_c) \geq
  c_j &\text{ for each } \hat{v_j} = \bar{q_{abc}} \\
  & x_i, c_i \in \{0, 1\}
\end{align*}
%
% \ar{Why are we representing bits as -1, 1 above and not in the integer program?
%   I don't see what the gain is ever talking about -1, 1 in this case...}
% \ar{Check the second constraint... It doesn't look right to me. I can satisfy it
%   by setting $d_j = 1, x_a = 0, x_b = x_c = 1$. Should probably be a factor of 3
%   in front of $d_j$. }
% \sw{I think this is fine since it is a disjunction.}
% \ar{Ah -- good point. The IP looks fine then. Its still confusing that we
%   introduce the variables as -1/1 for some reason, but then use them as 0/1}
% \sw{yes, fixed}
% \jh{We did this because for parities, $\pm 1$ is notationally simpler. Wanted to
%   maintain consistency.}
Note that the expressions $x_i, 1 - x_i$ encode the literals $x_i, \bar{x_i}$,
respectively, and the clause variable $c_i$ can be set to $1$ exactly when the
respective clause is satisfied.  Thus, the objective is the number of satisfied
clauses. The resulting integer program can be solved using any standard solver;
we use
\cplex.

\section{Case study: Parity queries}
In this section, we show how to apply \dq to another well-studied class of
queries: {\em parities}. Each specified by a subset $S$ of features, these
queries measure the number of records with an even number of bits on in $S$
compared to the number of records with an odd number of bits on in $S$.

\begin{definition}
  Let $\cX = \{-1, +1\}^d$. A {\em $k$-wise parity query} is a linear query
  specified by a subset of features $S \subseteq [d]$ with $|S| = k$, taking
  values
  \[
    q_S(x) = \left\{
      \begin{array}{ll}
        +1 &: \text{even number of } x_i = +1 \text{ for } i \in S \\
        -1 &: \text{otherwise.}
      \end{array}
    \right.
  \]
  Like before, we can define a {\em negated $k$-wise parity query}:
  \jh{Consider: call them even and odd parity queries instead? Maybe more
    intuitive.}
  \[
    \bar{q_S}(x) = \left\{
      \begin{array}{ll}
        +1 &: \text{odd number of } x_i = +1 \text{ for } i \in S \\
        -1 &: \text{otherwise.}
      \end{array}
    \right.
  \]
\end{definition}

% \citet{barak2007privacy} observed that answering $k$-way marginal
% queries can be reduced to answering $k$-wise parity queries to the same
% accuracy, in the following sense. \jh{Also cite Sasho's paper?}

% \begin{theorem}[\citet{barak2007privacy}]
%   \label{thm:mar-to-par}
%   Let $q_S$ be a $k$-way marginal query specified by the set of
%   features $S$, and let $D$ be a database of records. Then,
%   \[
%     q_S(D) = \frac{1}{2^k} \sum_{T \subseteq S}  p_T(D),
%   \]
%   where $p_T$ is the parity query for features $T$.
% \end{theorem}

% Note that the coefficients sum to $1$: hence, answering parity queries with
% additive error $\alpha$ is enough to answer marginal queries with additive error
% $\alpha$. We now consider how to handle these queries with our algorithm.
% Like marginal queries, it suffices
% to give the best-response optimization problem; unlike marginal queries, we need
% to handle $k$-wise parities for every $k \leq 3$ in order to apply
% \Cref{thm:mar-to-par}.

For the remainder, we specialize to $k = 3$.
Given sampled parity queries $\{ \hat{u_i} \}$ and negated parity queries $\{
\hat{v_i} \}$, the best response problem is to find the record $x \in \cX$ that
takes value $1$ on as many of these queries as possible. We can construct an
integer program for this task: introduce $d$ variables $x_i$, and two variables
$c_q, d_q$ for each sampled query. The following integer program encodes the
best-response problem.
\begin{align*}
  \max &\sum_i c_i & \\
  \text{such that } & \sum_j \sum_{p \in S_j} x_p = 2 d_i + c_i - 1 & \text{for each } \hat{u_i} = q_S \\
  & \sum_j \sum_{p \in S_j} x_p = 2 d_i + c_i & \text{for each } \hat{v_i} = \bar{q_S} \\
  & x_p, c_i \in \{0, 1\}, \quad d_i \in \{ 0, 1, 2 \} &
\end{align*}
Consider the (non-negated) parity queries first. The idea is that each variable
$c_i$ can be set to $1$ exactly when the corresponding parity query takes value
$1$, i.e., when $x$ has an even number of bits in $S$ set to $+1$. Since $|S|
\leq 3$, this even number will either be $0$ or $2$, hence is equal to $2d_i$
for $d_i \in \{0, 1\}$. A similar argument holds for the negated parity queries.

\begin{figure}[h]
  \centering
  \begin{center}
    \begin{tabular}{ | l || l | l |l |}
      \hline
      Dataset & Size   & Attributes & Binary attributes  \\
      \hline
      Adult   & 30162  & 14         & 235 \\
      KDD99   & 494021 & 41         & 396 \\
      Netflix & 480189 & 17,770     & 17,770 \\
      \hline
    \end{tabular}
  \end{center}
  \caption{Test Datasets}
  \label{fig:data}
\end{figure}

\section{Experimental evaluation}
\begin{figure*}[ht]
  \centering
  \subfloat{
    \includegraphics[width=0.325\textwidth]{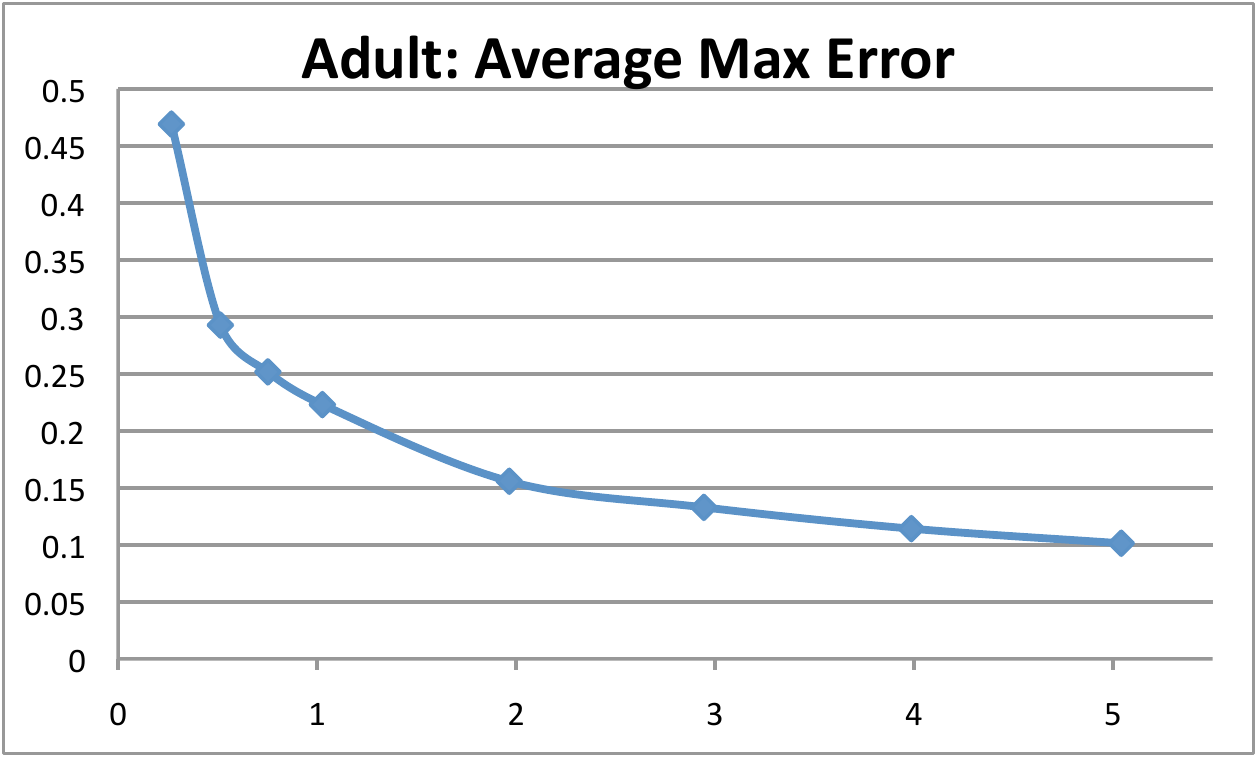}}
  \hfill
  \subfloat{
    \includegraphics[width=0.325\textwidth]{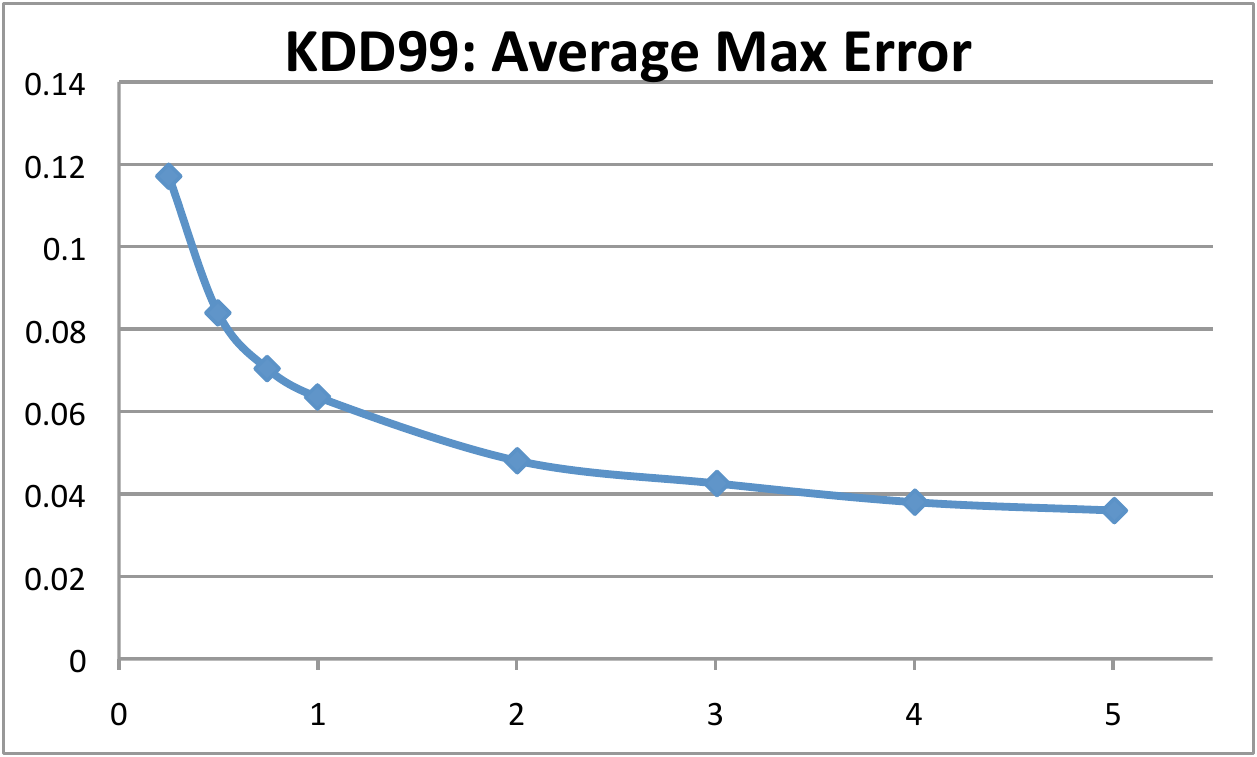}}
  \hfill
  \subfloat{
    \includegraphics[width=0.325\textwidth]{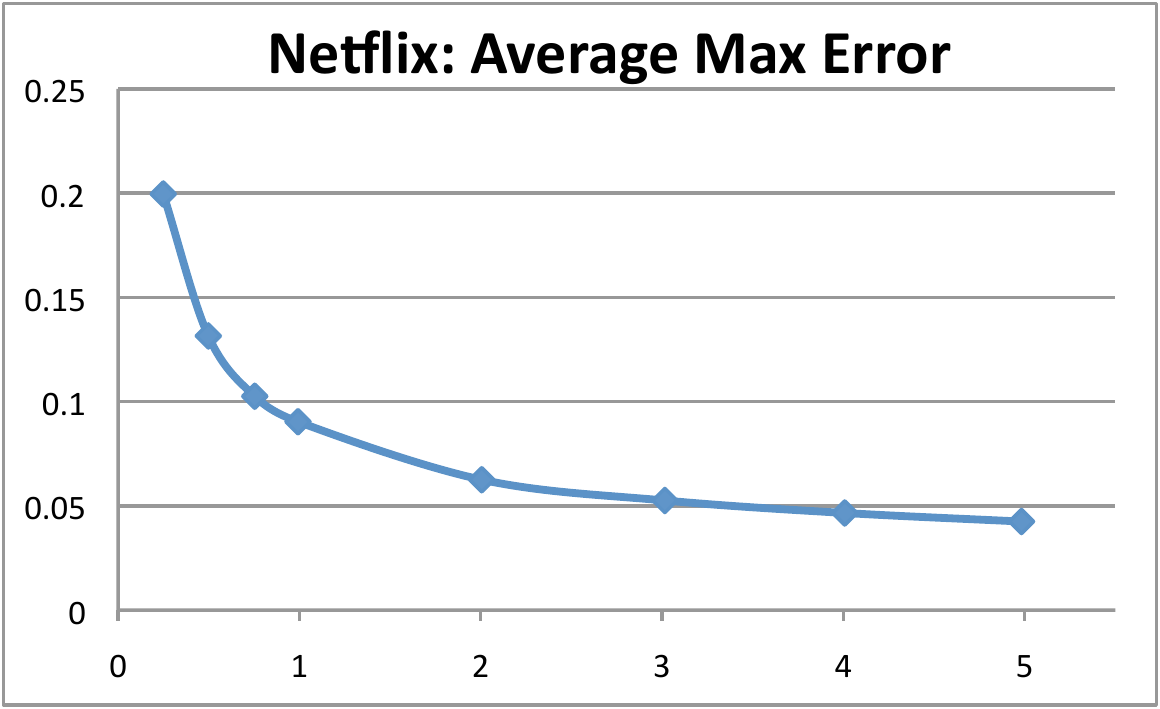}}

  \caption{Average max error of $(\eps, 0.001)$-private \dq on $500{,}000$
    $3$-way marginals versus $\eps$.}

  \label{fig:accuracy}
\end{figure*}
\begin{figure*}[ht]
  \centering
  \subfloat{
    \includegraphics[width=0.325\textwidth]{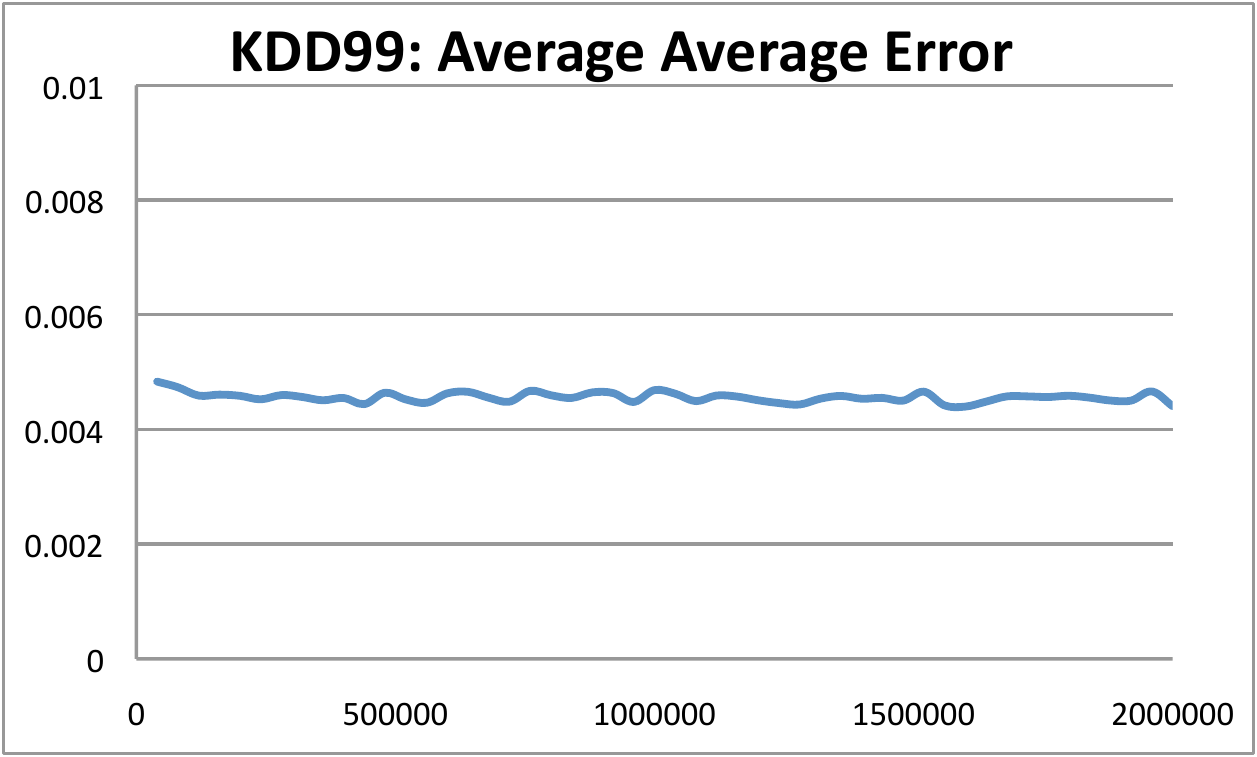}%
    \label{fig:q:kdd-avg}}
  \hfill
  \subfloat{
    \includegraphics[width=0.325\textwidth]{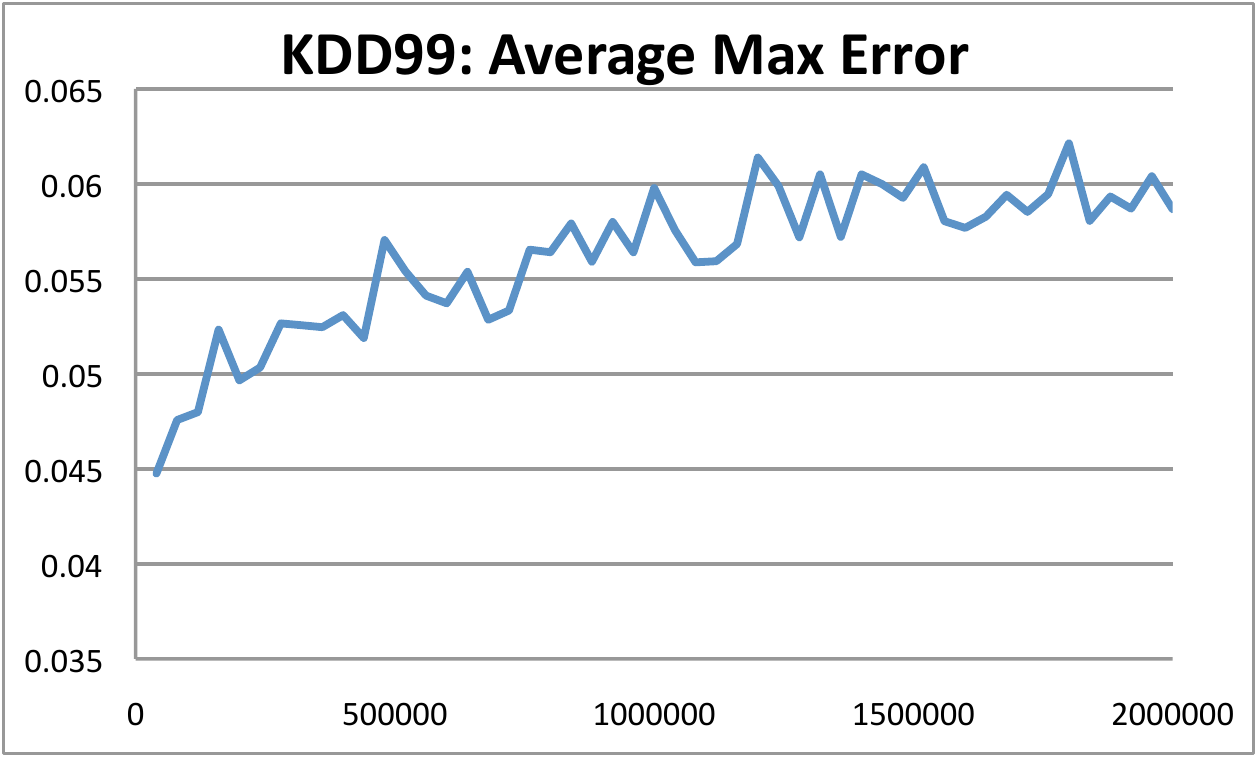}%
    \label{fig:q:kdd-max}}
  \hfill
  \subfloat{
    \includegraphics[width=0.325\textwidth]{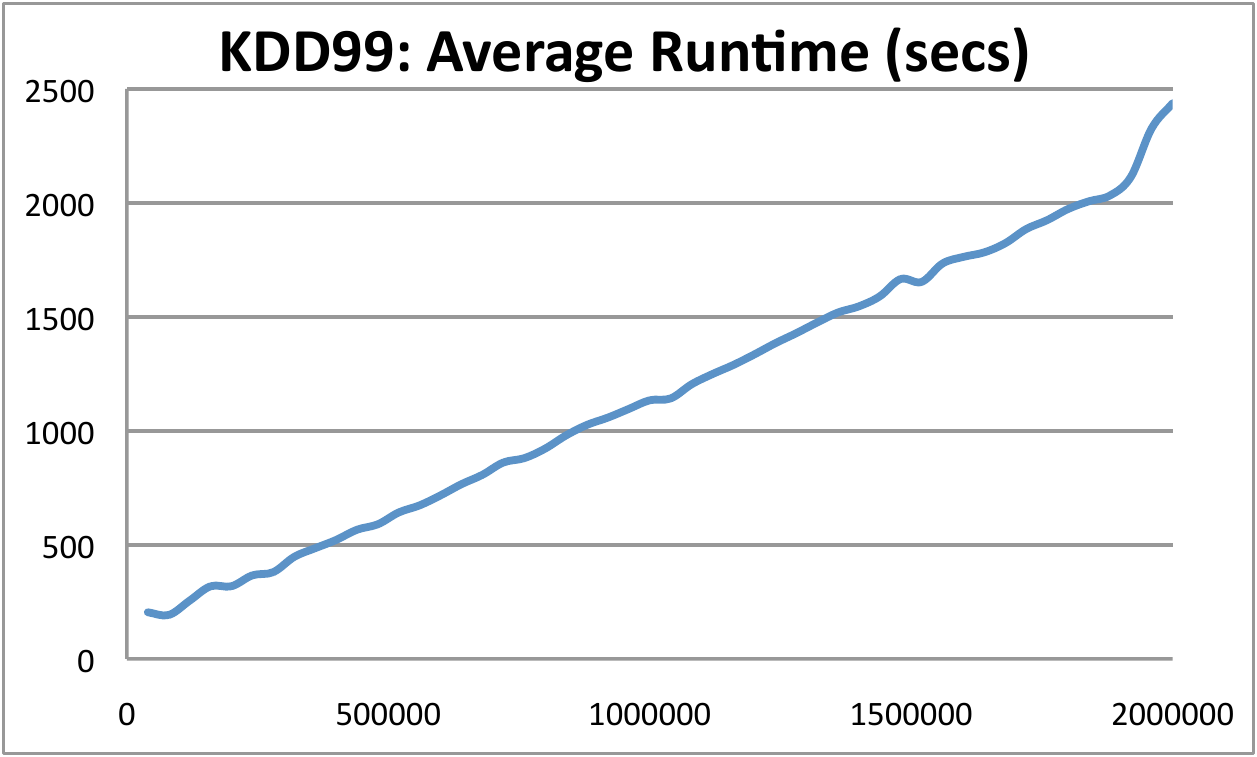}%
    \label{fig:q:kdd-rt}}

  \caption{Error and runtime of $(1, 0.001)$-private \dq on KDD99 versus number
    of queries.}
  \label{fig:queries}
\end{figure*}

\begin{figure*}[ht]
  \centering
  \subfloat{
    \includegraphics[width=0.325\textwidth]{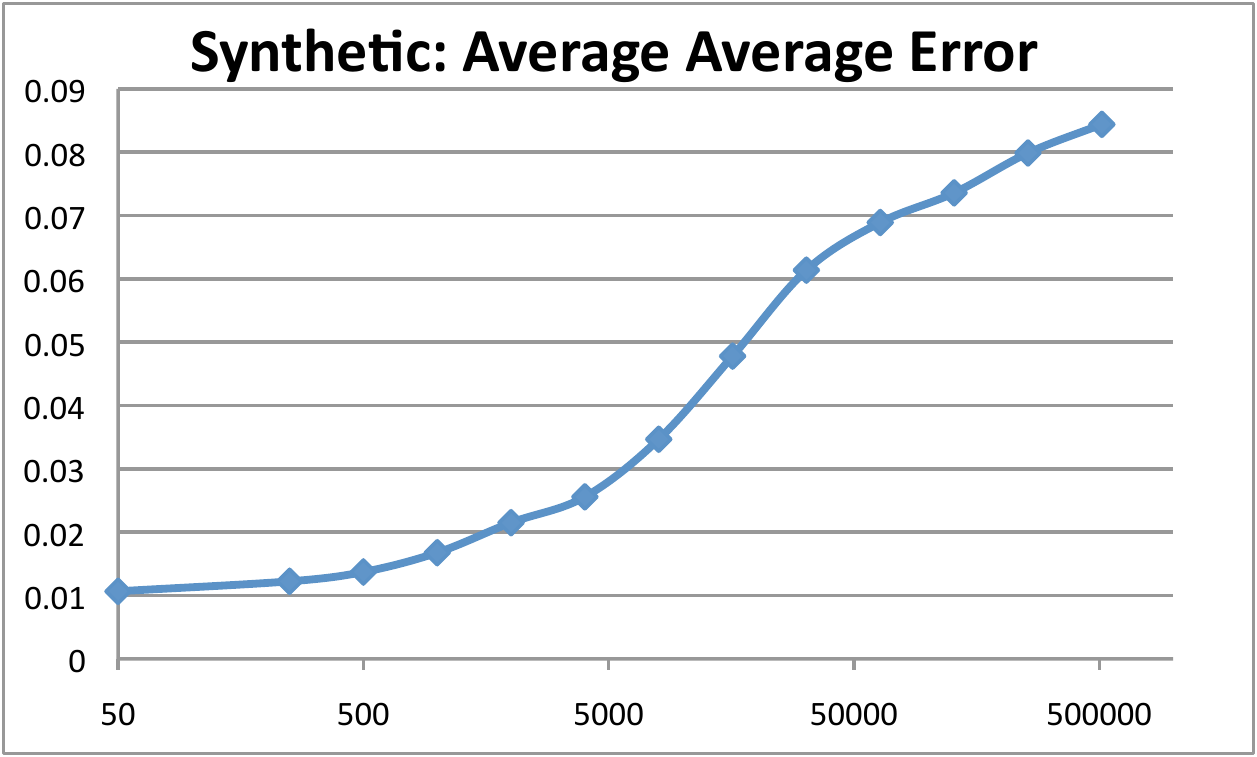}%
    \label{fig:att-error}}
  \hfill
  \subfloat{
    \includegraphics[width=0.325\textwidth]{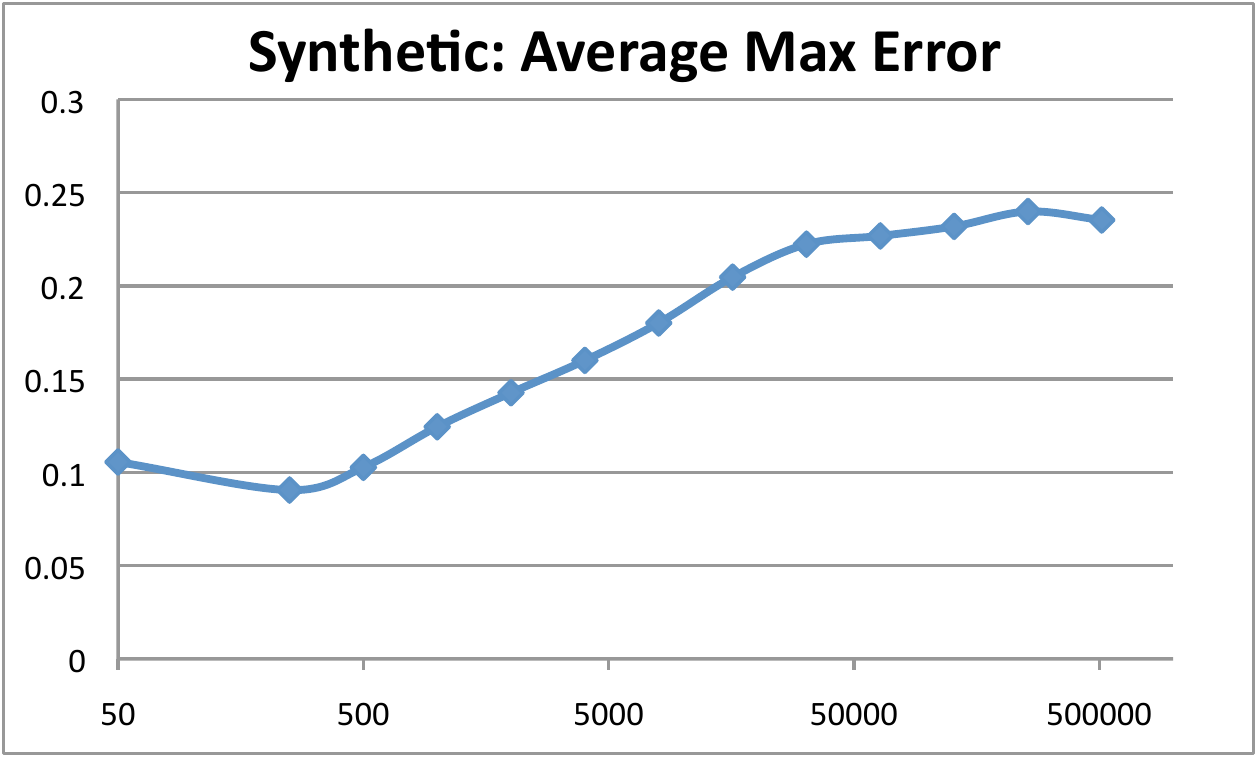}%
    \label{fig:att-max}}
  \hfill
  \subfloat{
    \includegraphics[width=0.325\textwidth]{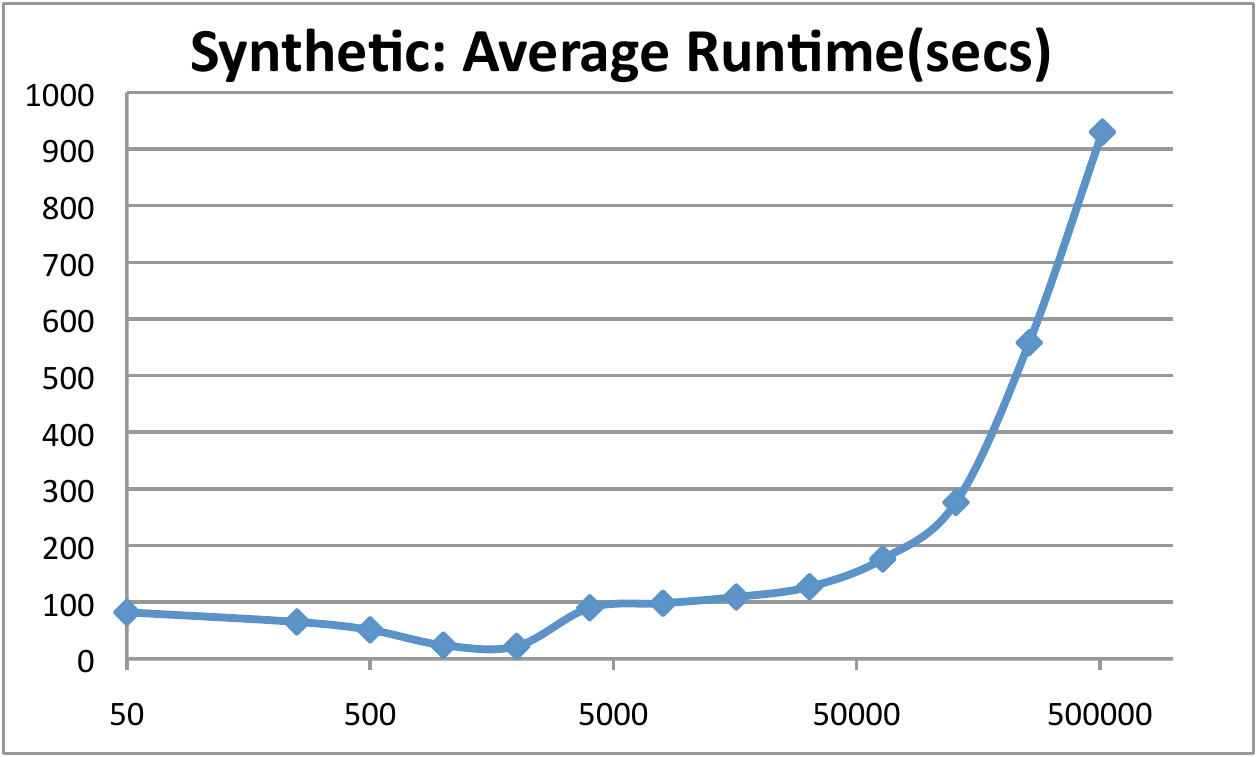}
    \label{fig:att-rt}}

  \caption{Error and runtime of $(1, 0.001)$-private \dq on $100{,}000$ $3$-way
    marginal queries versus number of attributes.}
    % \jh{Probably should remove the legend in the right graph.}
    % \sw{TODO: the plots would be further updated including this caption; also
    %   include max error}
    % \ar{Do we have an explanation for why the error is so much worse here than
    %   in our other experiments? Should we be using a larger value of $n$?  Also,
    %   the text of this section claims we scale the number of attributes to
    %   512,000, the caption says 128,000 -- we should be consistent (and
    %   correct!)}
    % \sw{our experiments crash with 512,000 attrs, will update accordingly.}
\label{fig:attributes}
\end{figure*}
% \fi

\begin{figure*}[h]
  \centering
  \subfloat{
    \includegraphics[width=0.35\textwidth]{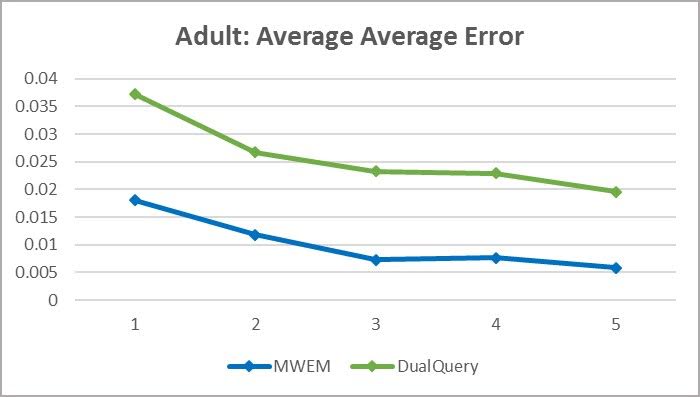}}
   %% \hfill
  \subfloat{
    \includegraphics[width=0.35\textwidth]{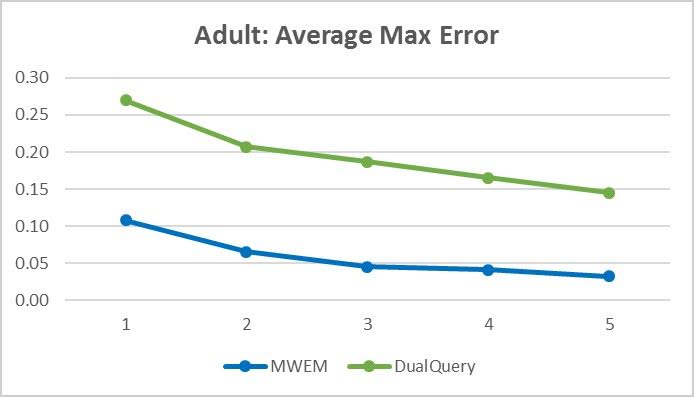}}

  %% \caption{ of $(\eps, 0.001)$-private \dq on $500{,}000$
  %%   $3$-way marginals versus $\eps$.}
  \caption{Comparison of the accuracy performance between \dq and
    \MWEM on a low-dimensional dataset with 17 attributes. Both
    algorithms answer 10,000 queries for the Adult dataset under
    $(\eps, 0)$-differential privacy, where $\eps$ ranges from 1 to
    5. In both plots, we show the average and maximum error as a
    function of the privacy parameter $\eps$.}
  \label{fig:mwem}
\end{figure*}

We evaluate \dq on a large collection of $3$-way marginal queries on
several real datasets (\Cref{fig:data}) and high dimensional synthetic
data. Adult (census data) and KDD99 (network packet data) are from the
UCI repository~\citep{uci}, and have a mixture of discrete (but
non-binary) and continuous attributes, which we discretize into binary
attributes. We also use the (in)famous \citet{netflix} movie ratings
dataset, with more than 17{,}000 binary attributes. More precisely, we
can consider each attribute (corresponding to a movie) to be $1$ if a
user has watched that movie, and $0$ otherwise.

Rather than set the parameters as in \Cref{alg:dualquery}, we
experiment with a range of parameters. For instance, we frequently run
for fewer rounds (lower $T$) and take fewer samples (lower $s$). As
such, the accuracy guarantee (\Cref{thm:utility}) need not hold for
our parameters. However, we find that our algorithm gives good error,
often much better than predicted. In all cases, our parameters satisfy
the privacy guarantee \Cref{thm:privacy}.

We will measure the accuracy of the algorithm using two
measures:~\emph{average error} and~\emph{max error}. Given a
collection of queries $Q=\{q_1, \ldots, q_k\}$, input database $D$ and
the synthetic database $\hat D$ output by our query release algorithm,
the average error is defined as
\[
\frac{1}{|Q|} \sum_j |q_j(D) - q_j(\hat D)|
\]
and the max error is defined as
\[
\max_j |q_j(D) - q_j(\hat D)|.
\]
We will run the algorithm multiple times and take the average of
both error statistics. While our algorithm uses ``negated'' $3$-way marginal
queries internally, all error figures are for normal, ``positive'' $3$-way
marginal queries only.

%\subsection*{Datasets}
% We evaluate our approach on several real datasets (\Cref{fig:data})
% and high dimensional synthetic data.
% Adult \iffull (census data) \fi and KDD99 \iffull (network
% packet data) \fi are from the UCI repository~\citep{uci},
% % \ar{The Adult citation looks funny})
% and have a mixture of discrete (but non-binary) and continuous attributes, which
% we discretize into binary attributes. We also use the (in)famous \citet{netflix}
% movie ratings dataset, with more than 17{,}000 binary attributes.\iffull{ More
%   precisely, we can consider each attribute (corresponding to a movie) to be $1$
%   if a user has watched that movie, and $0$ otherwise. }\fi{}
%
\eg{netflix is not exactly binary, it contains dates and rating, we
  make a binary version out of it.}

\eg{[A different try to the above two paragraphs, the first part talks
  about the queries and data, the second about the parameters, that
  IMHO, are not so important:]
% Dummy Comment as paragraph marker
  We evaluate \dq on large collections of $3$-way marginal queries
  over several real datasets (\Cref{fig:data}). Adult (census
  data) and KDD99 (network packet data) come from the
  UCI repository~\citep{uci}, while \citet{netflix} is the (in)famous
  movie ratings dataset. For our experiments, we discretize them into
  binary attributes.
% Dummy Comment as paragraph marker
  We have evaluated the algorithm with a wide variety of
  privacy-preserving parameters beyond the settings in
  \Cref{alg:dualquery}. For instance, we frequently run for fewer
  rounds and take less samples (lower $T$ and $s$). In that setting
  the accuracy guarantee (\Cref{thm:utility}) does not necessarily
  hold, however, we obtain good error, often much better than
  predicted.
}
\sw{Good try. Although this is not so much different from what's
  already written except the ordering.}

\subsection*{Accuracy} \label{sec:accuracy}
We evaluate the accuracy of the algorithm on $500{,}000$ $3$-way marginals on
Adult, KDD99 and Netflix. We report maximum error in \Cref{fig:accuracy},
averaged over $5$ runs.  (Marginal queries have range $[0, 1]$, so error $1$ is
trivial.) Again, all error figures are for normal, ``positive'' $3$-way marginal
queries only.
The runs are $(\eps, 0.001)$-differentially private, with $\eps$ ranging from
$0.25$ to $5$.\footnote{%
  Since our privacy analysis follows from
  \Cref{composition}, our algorithm actually satisfies $(\eps,
  \delta)$-privacy for smaller values of $\delta$. For example, our
  algorithm is also $(\sqrt{2}\eps, \delta')$-private for $\delta' =
  10^{-6}$. Similarly, we could choose any arbitrarily small value of
  $\delta$, and \Cref{composition} would tell us that our algorithm
  was $(\eps', \delta)$-differentially private for an appropriate
  value $\eps'$, which depends only sub-logarithmically on $1/\delta$.
}

For the Adult and KDD99 datasets, we set step size $\eta = 2.0$, sample size $s =
1000$ while varying the number of steps $T$ according to the privacy budget
$\eps$, using the formula from \Cref{thm:privacy}.  For the Netflix dataset, we
adopt the same heuristic except we set $s$ to be $5000$.

The accuracy improves noticeably when $\eps$ increases from $0.25$ to
$1$ across 3 datasets, and the improvement diminishes gradually with larger
$\eps$. With larger sizes, both KDD99 and Netflix datasets allow \dq to
run with more steps and get significantly better error.

% \jh{In general we should try to focus each section on a single aspect. This
%   section here is about accuracy versus epsilon; we can talk about runtime
%   separately.  This hopefully helps the reader figure out what to look at and
%   how to think about what our plots show, instead of wondering how what we are
%   talking about matches up with the figures we are telling them to look at.}
\subsection*{Scaling to More Queries}
Next, we evaluate accuracy and runtime when varying the number of queries. We
use a set of $40{,}000$ to $2$ million randomly generated marginals $\cQ$ on the
KDD99 dataset and run \dq with $(1, 0.001)$-privacy.
%; see \Cref{fig:queries} for plots.
% \sw{might explain meaningful queries in implementation}
% \jh{I think it could work to explain it in this section, maybe a few paragraphs
%   down.}
%
For all experiments, we use the same set of parameters: $\eta = 1.2, T = 170$
and $s = 1750$. By \Cref{thm:privacy}, each run of the experiment satisfies $(1,
0.001)$-differential privacy.  These parameters give stable performance as the
query class $\cQ$ grows.
As shown in \Cref{fig:queries}, both average and max error remain
mostly stable, demonstrating improved error compared to simpler
perturbation approaches.  For example, the Laplace mechanism's error
growth rate is $O(\sqrt{|\cQ|})$ under $(\eps, \delta)$-differential
privacy.
The runtime grows almost linearly in the number of queries, since we maintain a
distribution over all the queries.

\subsection*{Scaling to Higher Dimensional Data}
% \jh{I recommend the following template: first, explain the experiment at a high
%   level. What are the things on the axes, and why is this experiment
%   interesting? Say for this one: we want to how our algorithm performs on higher
%   dimensional databases. Next, present the plots.
%   Third, discussion. We should talk about what the plots show, and also include
%   more details about the experiments. Presenting too much detail early will make
%   it harder for the reader to follow what our experiment is doing; presenting
%   the plot before explaining the experiment may lead the reader to be totally
%   confused about what the plot is showing.}\sw{sounds good}
We also evaluate accuracy and runtime behavior for data dimension ranging
from $50$ to $512{,}000$.  We evaluate \dq under $(1, 0.001)$-privacy on
$100{,}000$ $3$-way marginals on synthetically generated datasets.  We report
runtime, max, and average error over $3$ runs in \Cref{fig:attributes}; note the
logarithmic scale for attributes axis. We do not include query evaluation in our
time measurements---this overhead is common to all approaches that answer a set
of queries.

When generating the synthetic data, one possibility is to set each attribute to
be $0$ or $1$ uniformly at random. However, this generates very uniform
synthetic data: a record satisfies any $3$-way marginal with probability $1/8$,
so most marginals will have value near $1/8$. To generate more challenging and
realistic data, we pick a separate bias $p_i \in [0, 1]$ uniformly at random for
each attribute $i$. For each data point, we then set attribute $i$ to be $1$
independently with probability equal to $p_i$.  As a result, different $3$-way
marginals have different answers on our synthetic data.

For parameters, we fix step size $\eta$ to be $0.4$, and increase the
sample size $s$ with the dimension of the data (from $200$ to
$50{,}000$) at the expense of running fewer steps. For these
parameters, our algorithm is $(1, 0.001)$-differentially private by
\Cref{thm:privacy}.  With this set of parameters, we are able to
obtain $8\%$ average error in an average of $10$ minutes of runtime,
excluding query evaluation.

\subsection*{Comparison with a Simple Baseline}

When interpreting our results, it is useful to compare them to a simple baseline
solution as a sanity check. Since our real data sets are sparse, we would
expect the answer of most queries to be small. A natural baseline is
the \emph{zeros data set}, where all records have all attributes set to $0$.

\jh{TODO.}\sw{put real-data discussion here}
For the experiments reporting max-error on real datasets, shown
in~\Cref{fig:accuracy}, the accuracy of~\dq out-performs the zeros
data set: the zeros data set has a max error of $1$ on both Adult
and KDD99 data sets, and a max error of $0.51$ on the Netflix
data set. For the experiments reporting average error in~\Cref{fig:queries}, the
zeros data set also has worse accuracy---it has an average error of
$0.11$ on the KDD dataset, whereas $\dq$ has an average error below $0.01$.

For synthetic data we can also compare to the zeros dataset, but it is a poor
benchmark since the dataset is typically not sparse. Hence, we also compare to a
\emph{uniform data set} containing one of each possible record.\footnote{%
  Note that this is the distribution that the MWEM algorithm starts with
  initially, so this corresponds to the error of MWEM after running it for 0
  rounds (running it for a non-zero number of rounds is not feasible given our
  data sizes)}
Recall that we generate our
synthetic data by first selecting a bias $p_i$ uniformly at random from $[0, 1]$
for each attribute $i$, then setting each record's attribute $i$ to be $1$ with
probability $p_i$ (independently across the attributes).

Given this data generation model, we can analytically compute the expected
average error on both benchmarks. Consider a query $q$ on
literals $a, b, c$.  The probability that $q$ evaluates to $1$ on a randomly
generated record is $p_a \cdot p_b \cdot p_c$. So, the expected value of $q$
when evaluated on the synthetic data is 
\[
  \mathbb{E} [ p_a \cdot p_b \cdot p_c ] = 0.125
\]
since $p_a, p_b, p_c$ are drawn independently and uniformly from $[0, 1]$. 

On the zeros data set, $q$ takes value $0$ unless it is of the form $\neg a
\land \neg b \land \neg c$, when it takes value $1$. If we average over all
queries $q$, this occurs on exactly $1/8$ of the queries. Thus, the expected
average error of the zeros database is
\[
  1/8 \cdot \mathbb{E} [ 1 - p_a \cdot p_b \cdot p_c ]
  +
  7/8 \cdot \mathbb{E} [ p_a \cdot p_b \cdot p_c ]
  =
  7/32 = 0.21875 .
\]
Similarly, the max-error will tend to $1$ as the size of the data set and number
of queries performed grows large. In practice, the max-error in our experiments
was above $0.98$.

We can perform a similar calculation with respect the uniform benchmark. Fix any
query $q$ on three literals $a,b,$ and $c$. Every query $q$ evaluates to exactly
$1/8$ on the uniform benchmark.  Thus, the expected error of the uniform
benchmark on $q$ is
\[
  \mathbb{E}[ | 1/8 - p_a \cdot p_b \cdot p_c | ] ,
\]
where $p_a, p_b, p_c$ are drawn independently and randomly from $[0, 1]$. This
error is
\[
  \mathbb{E}[ | 1/8 - p_a \cdot p_b \cdot p_c | ]
  =
  1/256 \cdot (7 + 18 \ln 2 + 2 (\ln 8)^3) \approx 0.11 ,
\]
by a direct calculation.

Similarly, the max-error will tend to $1 - 1/8 = 0.875$ as the size of the data
set and number of queries performed grows large. In practice, the max-error in
our experiments was above $0.85$. Thus, we see that \dq outperforms both of
these baseline comparisons on synthetic data.

\subsection*{Comparison with \MWEM}
Finally, we give a brief comparison between the algorithm \MWEM of
\citet{HLM12} and \dq in terms of the accuracy performance. We tested
both algorithms on the Adult dataset with a selection of $17$ attributes
to answer $10{,}000$ queries. The accuracy performance of \MWEM is better
than \dq on this low-dimensional dataset (shown
in~\Cref{fig:mwem}). Note that this matches with theoretical
guarantees---our theoretical accuracy bounds are worse than \MWEM as
shown by \citet{HR10,HLM12}. More generally, for low dimensional
datasets for which it is possible to maintain a distribution over
records, \MWEM likely remains the state of the art. Our work
complements \MWEM by allowing private data analysis on
higher-dimensional data sets. In this experiment, we fix the step size
of $\dq$ to be $\eta = 0.4$, and set the sample size and number of
steps to be $(s, T) = (35, 47), (40, 62), (45 , 71), (50, 78),(55,
83)$ to achieve privacy levels of $\eps = 1 , 2,3,4,5$
respectively. For the instantiation of \MWEM, we set the number of
steps to be $T = 15$ (see~\cite{HLM12} for more details of the
algorithm).

% Methodology:
\subsection*{Methodology}
In this section, we discuss our experimental setup in more detail.

\paragraph*{Implementation details.}
The implementation is written in OCaml, using the \cplex constraint
solver. We ran the experiments on a mid-range desktop
machine with a 4-core Intel Xeon processor and 12 Gb of
RAM. Heuristically, we set a timeout for each \cplex call to $20$
seconds, accepting the best current solution if we hit the
timeout. For the experiments shown, the timeout was rarely reached.
\eg{As we start to understand the oracles and the algorithm better, my
  guess is that in optimal settings we will hit the timeout almost
  always, if there is enough privacy budget, that is.
}
\paragraph*{Data discretization.}
We discretize KDD99 and Adult datasets into binary attributes by
mapping each possible value of a discrete attribute to a new binary
feature. We bucket continuous attributes, mapping each bucket to a new
binary feature. We also ensure that our randomly generated $3$-way marginal
queries are sensible (i.e., they don't require an original attribute to take two
different values).

\paragraph*{Setting free attributes.}
Since the collection of sampled queries may not involve all of the attributes,
\cplex often finds solutions that leave some attributes unspecified.
We set these {\em free} attributes heuristically:
%
% \eg{Umm, I wonder if a reader will understand that as that we perform
%   sparsity analysis.}
%
for real data, we set the attributes to $0$ as these datasets are fairly
sparse;%
%
% \eg{What does it mean for a dataset like adult to be sparse? We are
%   talking about the binarization, not the dataset itself!}
%
\footnote{%
  The adult and KDD99 datasets are sparse due to the way we discretize
  the data; for the Netflix dataset, most users have only viewed a
  tiny fraction of the $17{,}000$ movies.}
for synthetic data, we set attributes to $0$ or $1$ uniformly at random.%
\footnote{%
  For a more principled way to set these free
  attributes, the sparsity of the dataset could be estimated at a
  small additional cost to privacy.}
\eg{We are assuming as common cases --- for instance that cplex
  doesn't set most of the attributes --- cases that are not common
  when working under optimal parameters. IMHO we should be careful
  with that. I liked the old Oracle discussion more, IMHO it provided
  some structure to this problem.}
% \jh{We can say a bit more here. Did we find that how we set these attributes
%   really affects performance?}
% \sw{more discussion in the following paragraph}

\paragraph*{Parameter tuning.}
\dq has three parameters that can be set in a wide variety of configurations
without altering the privacy guarantee (\Cref{thm:privacy}): number of
iterations ($T$), number of samples ($s$), and learning rate ($\eta$), which
controls how aggressively to update the distribution. For a fixed level of
$\eps$ and $\delta$, there are many feasible private parameter settings.

Performance depends strongly on the choice of parameters: $T$ has an
obvious impact---runtime scales linearly with $T$. Other parameters have a more
subtle impact on performance---increasing $s$ increases the number of
constraints in
the integer program for \cplex, for example. We have investigated a range of
parameters, and for the experiments we have used informal
heuristics coming from our observations. We briefly describe some
of them here. For sparse datasets like the ones in \Cref{fig:data},
with larger $\eta$ (in the range of 1.5 to 2.0) and smaller $s$ (in
the same order as the number of attributes), we obtain good accuracy
when we set the free attributes to be $0$. But for denser data like
our synthetic data, we get better accuracy when we have smaller $\eta$
(around 0.4) and set the free attributes randomly. As for runtime, we
observe that \cplex could finish running quickly (in seconds) when the
sample size is in about the same range as the number of attributes
(factor of 2 or 3 different).

%% (parameters details deferred to our full version).
% \mg{I don't like the formulation of these last two sentences. They make it sound
%   like our work is in progress. Cannot we say instead something like: We have
%   experimented with parameters in different ranges adapting the ones giving the
%   best results with some heuristics coming from our observations.}

% \jh{Explain a bit, if we have some guidelines.}
% \sw{I shall not attempt to justify these heuristics. }

% Since data sparsity is sensitive information, this process
Ultimately, our parameters are chosen on a case-by-case basis for each data-set, which is not differentially private. Parameter setting should be done under differential privacy for a truly
realistic evaluation (which we do not do). Overall, we do not know of an approach that is both principled and practical to handle
this issue end-to-end; private parameter tuning is an area of active research (see e.g.,
\citet{CV13}).

\section{Discussion and conclusion}
We have given a new private query release mechanism that can handle datasets
with dimensionality multiple orders of magnitude larger than what was previously
possible.
%Our approach achieves accurate and private analysis on
%data that has dimensionality several orders of magnitude larger than previous
%approaches.
Indeed, it seems we have not reached the limits of our approach---even
on synthetic data with more than  $500{,}000$ attributes, \dq continues to
generate useful answers with about $30$ minutes of overhead on top of query
evaluation (which by itself is on the scale of hours).  We believe that \dq
makes private analysis of high dimensional data practical for the first time.

There is still plenty of room for further research. For example, can our
approach be improved to match the optimal theoretical accuracy guarantees for
query release? We do not know of any fundamental obstacles, but we do not know
how to give a better analysis of our current methods. The projection-based
approach of \citet{NTZ13,DNT14} seems promising---it achieves theoretically
optimal bounds, and like in our approach, the ``hard'' projection step does not
need to be solved privately. It deserves to be studied empirically. On the other
hand this approach does not produce synthetic data.
% \ar{Is this too strong? I recall ``11 minutes'' being thrown around earlier, but
%   if its hours please change...}
% \jh{I believe hours is including query evaluation time.}
\mg{Sorry for bringing back a thread of comments that seems
  closed. I'm not sure why do we need to talk about ``hours'' here. Or
  better this sounds like there is something wrong because we didn't
  mention about it in the experimental section but only here in the
  conclusion. Maybe we can put a sentence in the experimental section
  and then here just mention the overhead?}
\ar{The point here is that 30 minutes of overhead is very small, since it is a
  small fraction of the time you have to spend otherwise. If query evaluation
  took only seconds, then 30 minutes of overhead would be a long time.}
\sw{moved the last paragraph up in the comparison section}
\subsection*{Acknowledgments}
We thank Adam Smith, Cynthia Dwork and Ilya Mironov for productive discussions,
and for suggesting the Netflix dataset. This work was supported in part by NSF
grants CCF-1101389 and CNS-1065060.
Marco Gaboardi has been supported by the
European Community's Seventh Framework Programme FP7/2007-2013 under
grant agreement No. 272487.
This work was partially supported by a grant from the Simons Foundation (\#360368
to Justin Hsu).

\bibliographystyle{plainnat}
\bibliography{./refs}

\newcommand{\SortNoop}[1]{}
\begin{thebibliography}{38}
\providecommand{\natexlab}[1]{#1}
\providecommand{\url}[1]{\texttt{#1}}
\expandafter\ifx\csname urlstyle\endcsname\relax
  \providecommand{\doi}[1]{doi: #1}\else
  \providecommand{\doi}{doi: \begingroup \urlstyle{rm}\Url}\fi

\bibitem[Arora et~al.(2012)Arora, Hazan, and Kale]{AHK12}
Sanjeev Arora, Elad Hazan, and Satyen Kale.
\newblock
  \href{http://tocbeta.cs.uchicago.edu/articles/v008a006/v008a006.pdf}{The
  multiplicative weights update method: a meta-algorithm and applications}.
\newblock \emph{Theory of Computing}, 8\penalty0 (1):\penalty0 121--164, 2012.

\bibitem[Bache and Lichman(2013)]{uci}
K.~Bache and M.~Lichman.
\newblock \href{http://archive.ics.uci.edu/ml}{{UCI} machine learning
  repository}, 2013.

\bibitem[Beimel et~al.(2013)Beimel, Nissim, and Stemmer]{BNS13}
Amos Beimel, Kobbi Nissim, and Uri Stemmer.
\newblock \href{http://dl.acm.org/citation.cfm?id=2422450}{Characterizing the
  sample complexity of private learners}.
\newblock In \emph{{ACM} {SIGACT} {I}nnovations in {T}heoretical {C}omputer
  {S}cience (ITCS), Berkeley, California}, pages 97--110, 2013.

\bibitem[Blum et~al.(2013)Blum, Ligett, and Roth]{BLR08}
A.~Blum, K.~Ligett, and A.~Roth.
\newblock \href{http://arxiv.org/pdf/1109.2229v1}{A learning theory approach to
  noninteractive database privacy}.
\newblock \emph{Journal of the {ACM}}, 60\penalty0 (2):\penalty0 12, 2013.

\bibitem[Blum et~al.(2005)Blum, Dwork, Mc{S}herry, and Nissim]{BDMN05}
Avrim Blum, Cynthia Dwork, Frank Mc{S}herry, and Kobbi Nissim.
\newblock \href{http://research.microsoft.com/pubs/64351/bdmn.pdf}{Practical
  privacy: the sulq framework}.
\newblock In \emph{{ACM} {SIGACT--SIGMOD--SIGART} {S}ymposium on {P}rinciples
  of {D}atabase {S}ystems (PODS), Baltimore, Maryland}, pages 128--138, 2005.

\bibitem[Chaudhuri and Hsu(2011)]{CH11}
Kamalika Chaudhuri and Daniel Hsu.
\newblock
  \href{http://jmlr.org/proceedings/papers/v19/chaudhuri11a/chaudhuri11a.pdf}{Sample
  complexity bounds for differentially private learning}.
\newblock \emph{Journal of Machine Learning Research}, 19:\penalty0 155--186,
  2011.

\bibitem[Chaudhuri and Monteleoni(2008)]{CM08}
Kamalika Chaudhuri and Claire Monteleoni.
\newblock
  \href{http://books.nips.cc/papers/files/nips21/NIPS2008_0964.pdf}{Privacy-preserving
  logistic regression}.
\newblock In \emph{{C}onference on {N}eural {I}nformation {P}rocessing
  {S}ystems (NIPS), Vancouver, British Colombia}, pages 289--296, 2008.

\bibitem[Chaudhuri and Vinterbo(2013)]{CV13}
Kamalika Chaudhuri and Staal~A. Vinterbo.
\newblock
  \href{http://papers.nips.cc/paper/5014-a-stability-based-validation-procedure-for-differentially-private-machine-learning.pdf}{A
  stability-based validation procedure for differentially private machine
  learning}.
\newblock pages 2652--2660, 2013.

\bibitem[Chaudhuri et~al.(2011)Chaudhuri, Monteleoni, and Sarwate]{CMS11}
Kamalika Chaudhuri, Claire Monteleoni, and Anand~D. Sarwate.
\newblock
  \href{http://jmlr.org/papers/volume12/chaudhuri11a/chaudhuri11a.pdf}{Differentially
  private empirical risk minimization}.
\newblock \emph{Journal of Machine Learning Research}, 12:\penalty0 1069--1109,
  2011.

\bibitem[Chaudhuri et~al.(2012)Chaudhuri, Sarwate, and Sinha]{CSS12}
Kamalika Chaudhuri, Anand Sarwate, and Kaushik Sinha.
\newblock
  \href{http://books.nips.cc/papers/files/nips25/NIPS2012_0482.pdf}{Near-optimal
  differentially private principal components}.
\newblock In \emph{{C}onference on {N}eural {I}nformation {P}rocessing
  {S}ystems (NIPS), Lake Tahoe, California}, pages 998--1006, 2012.

\bibitem[Duchi et~al.(2013)Duchi, Jordan, and Wainwright]{DJW13}
J.C. Duchi, M.I. Jordan, and M.J. Wainwright.
\newblock
  \href{http://www.cs.berkeley.edu/~jduchi/projects/DuchiJoWa13_focs.pdf}{Local
  privacy and statistical minimax rates}.
\newblock In \emph{{IEEE} {S}ymposium on {F}oundations of {C}omputer {S}cience
  (FOCS), Berkeley, California}, 2013.

\bibitem[Dwork et~al.(2009)Dwork, Naor, Reingold, Rothblum, and
  Vadhan]{DNRRV09}
C.~Dwork, M.~Naor, O.~Reingold, G.N. Rothblum, and S.P. Vadhan.
\newblock \href{http://dl.acm.org/citation.cfm?id = 1536467}{On the complexity
  of differentially private data release: efficient algorithms and hardness
  results}.
\newblock In \emph{{ACM} {SIGACT} {S}ymposium on {T}heory of {C}omputing
  (STOC), Bethesda, Maryland}, pages 381--390, 2009.

\bibitem[Dwork et~al.(2010)Dwork, Rothblum, and Vadhan]{DRV10}
C.~Dwork, G.N. Rothblum, and S.~Vadhan.
\newblock
  \href{http://research.microsoft.com/pubs/155170/dworkrv10.pdf}{Boosting and
  differential privacy}.
\newblock In \emph{{IEEE} {S}ymposium on {F}oundations of {C}omputer {S}cience
  (FOCS), Las Vegas, Nevada}, pages 51--–60, 2010.

\bibitem[Dwork et~al.(2006)Dwork, McSherry, Nissim, and Smith]{DMNS06}
Cynthia Dwork, Frank McSherry, Kobbi Nissim, and Adam Smith.
\newblock \href{http://dx.doi.org/10.1007/11681878_14}{Calibrating noise to
  sensitivity in private data analysis}.
\newblock In \emph{{IACR} {T}heory of {C}ryptography {C}onference (TCC), New
  York, New York}, pages 265--284, 2006.

\bibitem[Dwork et~al.(2014)Dwork, Nikolov, and Talwar]{DNT14}
Cynthia Dwork, Aleksandar Nikolov, and Kunal Talwar.
\newblock \href{http://arxiv.org/abs/1308.1385}{Using convex relaxations for
  efficiently and privately releasing marginals}.
\newblock In \emph{{SIGACT}--{SIGGRAPH} {S}ymposium on {C}omputational
  {G}eometry (SOCG), Kyoto, Japan}, page 261, 2014.

\bibitem[Freund and Schapire(1996)]{FS96}
Y.~Freund and R.E. Schapire.
\newblock \href{http://dl.acm.org/citation.cfm?id=238163}{Game theory, on-line
  prediction and boosting}.
\newblock In \emph{{C}onference on {C}omputational {L}earning {T}heory
  ({CoLT}), Desenzano sul Garda, Italy}, pages 325--332, 1996.

\bibitem[Gupta et~al.(2012)Gupta, Roth, and Ullman]{GRU12}
A.~Gupta, A.~Roth, and J.~Ullman.
\newblock \href{http://arxiv.org/pdf/1107.3731v2}{Iterative constructions and
  private data release}.
\newblock In \emph{{IACR} {T}heory of {C}ryptography {C}onference (TCC),
  Taormina, Italy}, pages 339--356, 2012.

\bibitem[Gupta et~al.(2013)Gupta, Hardt, Roth, and Ullman]{GHRU13}
Anupam Gupta, Moritz Hardt, Aaron Roth, and Jonathan Ullman.
\newblock \href{http://epubs.siam.org/doi/abs/10.1137/110857714}{Privately
  releasing conjunctions and the statistical query barrier}.
\newblock \emph{SIAM Journal on Computing}, 42\penalty0 (4):\penalty0
  1494--1520, 2013.

\bibitem[Hardt and Rothblum(2010)]{HR10}
Moritz Hardt and Guy~N. Rothblum.
\newblock \href{http://www.mit.edu/~rothblum/papers/pmw.pdf}{A multiplicative
  weights mechanism for privacy-preserving data analysis}.
\newblock In \emph{{IEEE} {S}ymposium on {F}oundations of {C}omputer {S}cience
  (FOCS), Las Vegas, Nevada}, pages 61--70, 2010.

\bibitem[Hardt et~al.(2012)Hardt, Ligett, and {McSherry}]{HLM12}
Moritz Hardt, Katrina Ligett, and Frank {McSherry}.
\newblock \href{http://arxiv.org/pdf/1012.4763v1}{A simple and practical
  algorithm for differentially private data release}.
\newblock In \emph{{C}onference on {N}eural {I}nformation {P}rocessing
  {S}ystems (NIPS), Lake Tahoe, California}, pages 2348--2356, 2012.

\bibitem[Hsu et~al.(2013)Hsu, Roth, and Ullman]{HRU13}
Justin Hsu, Aaron Roth, and Jonathan Ullman.
\newblock \href{http://arxiv.org/pdf/1211.0877v2}{Differential privacy for the
  analyst via private equilibrium computation}.
\newblock In \emph{{ACM} {SIGACT} {S}ymposium on {T}heory of {C}omputing
  (STOC), Palo Alto, California}, pages 341--350, 2013.

\bibitem[Kasiviswanathan et~al.(2011)Kasiviswanathan, Lee, Nissim,
  Raskhodnikova, and Smith]{KLNRS08}
Shiva~Prasad Kasiviswanathan, Homin~K. Lee, Kobbi Nissim, Sofya Raskhodnikova,
  and Adam Smith.
\newblock \href{http://arxiv.org/pdf/0803.0924v3.pdf}{What can we learn
  privately?}
\newblock \emph{SIAM Journal on Computing}, 40\penalty0 (3):\penalty0 793--826,
  2011.

\bibitem[Kearns(1998)]{Kearns98}
Michael~J. Kearns.
\newblock \href{http://doi.acm.org/10.1145/293347.293351}{Efficient
  noise-tolerant learning from statistical queries}.
\newblock \emph{Journal of the {ACM}}, 45\penalty0 (6):\penalty0 983--1006,
  1998.

\bibitem[Kifer et~al.(2012)Kifer, Smith, and Thakurta]{KST12}
Daniel Kifer, Adam Smith, and Abhradeep Thakurta.
\newblock
  \href{http://jmlr.org/proceedings/papers/v23/kifer12/kifer12.pdf}{Private
  convex empirical risk minimization and high-dimensional regression}.
\newblock \emph{Journal of Machine Learning Research}, 1:\penalty0 41, 2012.

\bibitem[Li and Miklau(2012)]{LM12}
Chao Li and Gerome Miklau.
\newblock \href{http://arxiv.org/pdf/1202.3807v1}{An adaptive mechanism for
  accurate query answering under differential privacy}.
\newblock volume~5, pages 514--525, 2012.

\bibitem[Li et~al.(2010)Li, Hay, Rastogi, Miklau, and Mc{G}regor]{CHRMM10}
Chao Li, Michael Hay, Vibhor Rastogi, Gerome Miklau, and Andrew Mc{G}regor.
\newblock \href{http://arxiv.org/pdf/0912.4742v2}{Optimizing linear counting
  queries under differential privacy}.
\newblock In \emph{{ACM} {SIGACT--SIGMOD--SIGART} {S}ymposium on {P}rinciples
  of {D}atabase {S}ystems (PODS), Indianapolis, Indiana}, pages 123--134, 2010.

\bibitem[McSherry and Talwar(2007)]{MT07}
F.~McSherry and K.~Talwar.
\newblock
  \href{http://doi.ieeecomputersociety.org/10.1109/FOCS.2007.41}{Mechanism
  design via differential privacy}.
\newblock In \emph{{IEEE} {S}ymposium on {F}oundations of {C}omputer {S}cience
  (FOCS), Providence, Rhode Island}, pages 94--103, 2007.

\bibitem[Narayanan and Shmatikov(2008)]{NV08}
A.~Narayanan and V.~Shmatikov.
\newblock \href{http://arxiv.org/pdf/cs/0610105.pdf}{Robust de-anonymization of
  large sparse datasets}.
\newblock In \emph{{IEEE} {S}ymposium on {S}ecurity and {P}rivacy (S\&P),
  Oakland, California}, pages 111--125, 2008.

\bibitem[Netflix()]{netflix}
Netflix.
\newblock \href{http://www.netflixprize.com/}{Netflix prize}.

\bibitem[Nikolov et~al.(2013)Nikolov, Talwar, and Zhang]{NTZ13}
Aleksandar Nikolov, Kunal Talwar, and Li~Zhang.
\newblock \href{http://arxiv.org/abs/1212.0297}{The geometry of differential
  privacy: the sparse and approximate cases}.
\newblock In \emph{{ACM} {SIGACT} {S}ymposium on {T}heory of {C}omputing
  (STOC), Palo Alto, California}, pages 351--360, 2013.

\bibitem[Roth and Roughgarden()]{RR10}
Aaron Roth and Tim Roughgarden.
\newblock \href{http://arxiv.org/pdf/0911.1813}{Interactive privacy via the
  median mechanism}.
\newblock In \emph{{ACM} {SIGACT} {S}ymposium on {T}heory of {C}omputing
  (STOC), Cambridge, Massachusetts}, pages 765--774.

\bibitem[Rubinstein et~al.(2012)Rubinstein, Bartlett, Huang, and Taft]{RBHT09}
Benjamin I.~P. Rubinstein, Peter~L. Bartlett, Ling Huang, and Nina Taft.
\newblock
  \href{http://repository.cmu.edu/cgi/viewcontent.cgi?article=1065&context=jpc}{Learning
  in a large function space: Privacy-preserving mechanisms for {SVM} learning}.
\newblock \emph{Journal of Privacy and Confidentiality}, 4\penalty0
  (1):\penalty0 4, 2012.

\bibitem[Thakurta and Smith(2013)]{TS13}
Abhradeep~G. Thakurta and Adam Smith.
\newblock
  \href{http://media.nips.cc/nipsbooks/nipspapers/paper_files/nips26/1270.pdf}{({N}early)
  optimal algorithms for private online learning in full-information and bandit
  settings}.
\newblock In \emph{{C}onference on {N}eural {I}nformation {P}rocessing
  {S}ystems (NIPS), Lake Tahoe, California}, pages 2733--2741, 2013.

\bibitem[Thaler et~al.(2012)Thaler, Ullman, and Vadhan]{TUV12}
Justin Thaler, Jonathan Ullman, and Salil Vadhan.
\newblock \href{http://arxiv.org/abs/1205.1758}{Faster algorithms for privately
  releasing marginals}.
\newblock In \emph{International Colloquium on Automata, Languages and
  Programming (ICALP), Warwick, England}, pages 810--821, 2012.

\bibitem[Ullman(2013)]{Ullman13}
J.~Ullman.
\newblock \href{http://arxiv.org/pdf/1207.6945v3}{Answering $n^{2+ o(1)}$
  counting queries with differential privacy is hard}.
\newblock In \emph{{ACM} {SIGACT} {S}ymposium on {T}heory of {C}omputing
  (STOC), Palo Alto, California}, pages 361--370, 2013.

\bibitem[Ullman and Vadhan(2011)]{UV11}
J.~Ullman and S.P. Vadhan.
\newblock
  \href{http://eccc.hpi-web.de/report/2010/017/revision/2/download}{{PCPs} and
  the hardness of generating private synthetic data}.
\newblock In \emph{{IACR} {T}heory of {C}ryptography {C}onference (TCC),
  Providence, Rhode Island}, pages 400--416, 2011.

\bibitem[Yaroslavtsev et~al.(2013)Yaroslavtsev, Cormode, Procopiuc, and
  Srivastava]{CPSY12}
Grigory Yaroslavtsev, Graham Cormode, Cecilia~M. Procopiuc, and Divesh
  Srivastava.
\newblock
  \href{http://doi.ieeecomputersociety.org/10.1109/ICDE.2013.6544871}{Accurate
  and efficient private release of datacubes and contingency tables}.
\newblock In \emph{{IEEE} {I}nternational {C}onference on {D}ata {E}ngineering
  (ICDE), Brisbane, Australia}, pages 745--756, 2013.

\bibitem[Zhang et~al.(2014)Zhang, Cormode, Procopiuc, Srivastava, and
  Xiao]{PrivBayes}
Jun Zhang, Graham Cormode, Cecilia~M. Procopiuc, Divesh Srivastava, and Xiaokui
  Xiao.
\newblock
  \href{http://dimacs.rutgers.edu/~graham/pubs/papers/PrivBayes.pdf}{Privbayes:
  Private data release via bayesian networks}.
\newblock In \emph{{ACM} {SIGMOD} {I}nternational {C}onference on {M}anagement
  of {D}ata (SIGMOD), Snowbird, Utah}, 2014.

\end{thebibliography}

\end{document}